\documentclass[10pt,journal,twoside]{IEEEtran}

\usepackage{amsmath,amssymb,amsthm,bm}
\usepackage{mathtools}
\usepackage{graphicx}
\usepackage{cite}
\usepackage{microtype}
\usepackage{balance}

\newcommand{\R}{\mathbb{R}}
\newcommand{\E}{\mathbb{E}}
\newcommand{\diff}{\mathrm{d}}
\newcommand{\norm}[1]{\left\lVert#1\right\rVert}

\theoremstyle{plain}
\newtheorem{theorem}{Theorem}
\newtheorem{lemma}{Lemma}
\newtheorem{definition}{Definition}
\newtheorem{proposition}{Proposition}
\newtheorem{corollary}{Corollary}

\theoremstyle{remark}
\newtheorem{remark}{Remark}

\usepackage{booktabs}
\usepackage{tabularx}
\usepackage{array}
\newcolumntype{Y}{>{\raggedright\arraybackslash}X}

\usepackage{tikz}
\usetikzlibrary{arrows.meta, decorations.pathmorphing, calc, patterns}

\usepackage[
    colorlinks=true,
    linkcolor=blue,
    citecolor=blue,
    urlcolor=blue,
    bookmarks=true,
    pdfborder={0 0 0}
]{hyperref}

\usepackage[all]{hypcap}
\usepackage[capitalize]{cleveref}

\title{Capacity Scaling Laws for Boundary-Induced Drift--Diffusion Noise Channels}

\author{%
Yen-Chi~Lee,~\IEEEmembership{Member,~IEEE}%
\thanks{This work was supported by the National Science and Technology Council of Taiwan (NSTC 113-2115-M-008-013-MY3). (Corresponding author: Yen-Chi Lee.)}%
\thanks{%
Y.-C. Lee is with the Department of Mathematics, National Central University, Taoyuan, Taiwan (e-mail: \texttt{yclee@math.ncu.edu.tw}).%
}%
}

\begin{document}
\maketitle


\begin{abstract}
This paper studies the high-power capacity scaling of additive noise channels
whose noise arises from the first-hitting location of a multidimensional
drift--diffusion process on an absorbing hyperplane.
By identifying the underlying stochastic transport mechanism as a Gaussian
variance--mixture, we introduce and analyze the \emph{Normally-Drifted
First-Hitting Location} ($\mathrm{NDFHL}$) family as a geometry-driven model for
boundary-induced noise.
Under a second-moment constraint, we derive an exact high-SNR capacity expansion
and show that the asymptotic upper and lower bounds coincide at the constant
level, yielding a vanishing capacity gap.
As a consequence, isotropic Gaussian signaling is asymptotically
capacity-achieving for all fixed drift strengths, despite the non-Gaussian and
semi-heavy-tailed nature of the noise.
The pre-log factor is determined solely by the dimension of the receiving
boundary, revealing a geometric origin of the channel’s degrees of freedom.
The refined expansion further uncovers an entropy-dominant universality, whereby
all physical parameters of the transport process—including drift strength,
diffusion coefficient, and boundary separation—affect the capacity only through
the differential entropy of the induced noise.
Although the $\mathrm{NDFHL}$ density does not admit a simple closed form, its
entropy is shown to be finite and to vary continuously as the drift vanishes,
thereby connecting the finite-variance regime with the singular
infinite-variance Cauchy limit.
Together, these results provide a unified geometric and information-theoretic
characterization of boundary-hitting channels across both regular and singular
transport regimes.
\end{abstract}



\section{Introduction}
\label{sec:introduction}

\IEEEPARstart{F}{irst-passage} phenomena serve as a foundational mechanism in nonequilibrium statistical physics \cite{zwanzig2001nonequilibrium}, reaction-diffusion theory \cite{murray2007mathematical}, and, increasingly, in emerging communication paradigms where information is carried by physical particles \cite{Akyildiz:2008,Nakano:2013,Farsad:2016}. While classical stochastic analysis often focuses on temporal statistics \cite{salminen1988first,srinivas2012molecular,grebenkov2020joint,mazzolo2024first}, the \emph{first-hitting location} (FHL) on a spatially extended boundary \cite{byczkowski2010hitting} constitutes a critical geometric observable—effectively serving as the output signal (i.e., observable) of a transport-based channel \cite{lee26FHL}. From this perspective, the channel noise is not an extrinsic impairment but an intrinsic mechanism of the transport dynamics. 
Specifically, the interplay between drift and diffusion defines a characteristic length scale $l_u = 1/u$ (also referred to as the characteristic propagation distance \cite{lee25tail}), which governs the statistical nature of the noise: strong drift confines the spatial fluctuations, whereas vanishing drift leads to scale-free, heavy-tailed behavior \cite{lee25tail}. 
Understanding how this physical length scale regularizes the boundary statistics is therefore a prerequisite for establishing the fundamental information-theoretic limits of such geometry-dependent physical channels.

Beyond physical transport, first-hitting events also play a fundamental role in
financial risk analysis and mathematical finance
\cite{Merton:1974,Asmussen:2000}.
In multi-dimensional asset models, the crossing of a prescribed boundary—such as
a default barrier, solvency threshold, or regulatory constraint—marks a critical
system event.
At this instant, the first-hitting location encodes the joint state of correlated
assets, stochastic volatility factors, or auxiliary risk variables
\cite{BlackCox:1976,Jeanblanc:2009}.
From an information-theoretic perspective, the boundary observable thus represents
the residual uncertainty of the system conditioned on the occurrence of failure
or exit, rather than merely the timing of the event.
Characterizing the statistical structure and entropy of such boundary-induced
observables is therefore central to quantifying risk, dependence, and information
flow in path-dependent financial systems.

Despite the diverse physical and financial manifestations of first-hitting time--location phenomena, a unified information-theoretic framework for characterizing these boundary observables remained elusive until recently.

The formalization of these stochastic transport dynamics into a communication-theoretic framework was initiated in \cite{lee2016distribution,pandey2018molecular}. More recently, motivated by the unique constraints of molecular communication (MC), Lee et al. \cite{lee24CFAP} introduced the concrete \emph{First-Arrival-Position} (FAP) channel. In this model, the associated noise distribution was termed \emph{Vertically-Drifted} FAP (VDFAP), representing the spatial uncertainty of particles at the moment of capture. While that study succeeded in deriving the exact probability density function (PDF) of the hitting location for a drift--diffusion process, the resulting expression relies on complicated modified Bessel functions that obscure the underlying geometric simplicity.
Consequently, two critical theoretical gaps remain. First, the high-SNR capacity scaling law was not explicitly characterized in because the lower bound was expressed using the differential entropy of VDFAP noise, which lacks a general closed-form expression. 
Second, and more fundamentally, the density-based analysis fails in the regime of vanishing drift ($u \to 0$). As the noise power diverges in this limit, the previously established Max-Ent upper bound explodes, leaving the capacity scaling of the heavy-tailed Cauchy-limit regime analytically unresolved despite the physical continuity of the underlying transport process.

In this paper, we resolve these ambiguities by shifting the analytical focus from the intractable PDF to the underlying stochastic generation mechanism. Our core insight is identifying the boundary-induced noise as a \emph{Gaussian Variance-Mean Mixture} (GVM) \cite{BarndorffNielsen:1982} (in particular, the zero-mean-mixing setting), where the mixing variable is the inverse-Gaussian (IG) distributed first-hitting time \cite{Chhikara:1989}. This perspective reveals that the complicated Bessel-type semi-heavy tails (i.e., exhibiting exponential decay slower than Gaussian) are merely manifestations of the random arrival times modulating the spatial diffusion. 
By conditioning on the latent temporal dynamics, we convert the non-Gaussian channel into a conditionally Gaussian one. This enables the use of entropy inequalities and asymptotic expansions to derive matching bounds—a paradigm that has proven highly effective for characterizing the high-SNR behavior of complex physical channels, ranging from flat fading channels \cite{Lapidoth:2003} to wireless optical intensity channels \cite{moser2018}. 
Our GVM approach not only simplifies the mathematical treatment but also unifies the finite-variance and infinite-variance regimes under a single geometric framework.

The main contributions of this work are summarized as follows.

\begin{itemize}
    \item \textbf{Boundary-Induced Noise as a Gaussian Variance Mixture:}
    We formalize the \emph{Normally-Drifted First-Hitting Location}
    ($\mathrm{NDFHL}^{(d)}$) family as a geometry-induced Gaussian
    variance--mixture, parameterized by the ambient dimension $d$, a separation
    distance $\lambda$, and a normalized drift speed $u$.
    A compact closed-form characteristic function is obtained, and the infinite
    divisibility of the family is established, reflecting the additivity of
    layered stochastic transport mechanisms.

    \item \textbf{Geometry-Driven Capacity Scaling:}
    For additive channels corrupted by $\mathrm{NDFHL}$ noise under a
    second-moment constraint, we derive an exact high-SNR capacity expansion.
    The pre-log factor is shown to be determined solely by the dimension of the
    receiving boundary, demonstrating that the fundamental degrees of freedom
    are geometric in origin and insensitive to the tail behavior of the noise.

    \item \textbf{Vanishing Constant Gap and Asymptotic Optimality:}
    A refined high-SNR analysis reveals that the asymptotic upper and lower
    capacity bounds coincide at the constant level, yielding a vanishing
    capacity gap.
    As a consequence, isotropic Gaussian signaling is asymptotically
    capacity-achieving for all fixed drift strengths $u>0$, despite the
    non-Gaussian and semi-heavy-tailed nature of the boundary-induced noise.

    \item \textbf{Continuity to the Cauchy Regime:}
    In the vanishing-drift limit $u\to 0$, the $\mathrm{NDFHL}$ family converges
    in distribution to the multivariate Cauchy law.
    While the differential entropy $h(\mathbf N)$ does not admit a tractable
    closed-form expression for general $\mathrm{NDFHL}$ noise---and remains
    analytically intricate even in the Cauchy limit (see \eqref{eq:gp_def})---its
    finiteness is established rigorously.
    Numerical evaluations are used solely to visualize the entropy behavior as a
    function of $u$, confirming that although the noise variance diverges, the
    differential entropy---and hence the high-SNR capacity offset---remains
    finite and continuous, consistent with recent results on additive Cauchy
    noise~\cite{pang2025information}.
\end{itemize}

The remainder of this paper is organized as follows.
Section~\ref{sec:previous_work} reviews the FAP channel framework.
Section~\ref{sec:stochastic_foundation} establishes the stochastic foundation
and Gaussian variance--mixture structure.
The $\mathrm{NDFHL}$ distribution family is defined in
Section~\ref{sec:NDFHL_distribution}, followed by its analytical properties in
Section~\ref{sec:properties}.
The capacity scaling analysis and refined high-SNR expansion are presented in
Sections~\ref{sec:scaling} and~\ref{subsec:refined_scaling}.
Section~\ref{sec:singular_limit} discusses the singular Cauchy limit, and
Section~\ref{sec:conclusion} concludes the paper.


\section{Previous Work: The FAP Channel Framework}
\label{sec:previous_work}

The information-theoretic investigation of boundary-induced hitting-location observables was formally initiated in \cite{lee24CFAP}. That work introduced the FAP channel and provided the first rigorous characterization of the noise statistics induced by the first-hitting of a drift--diffusion process on a planar absorbing boundary. In this section, we summarize the analytical formulations derived in \cite{lee24CFAP}, which serve as the 
background of the current work.

\subsection{Explicit Noise Characterization}
\label{subsec:explicit_char}

Consider a particle released in an Euclidean space $\mathbb{R}^{p+1}$ (we set $d=p+1\in\mathbb{N}-\{1\}$) subject to a constant vertical drift $u > 0$ (towards the boundary) and diffusion coefficient $\sigma>0$. The noise vector $\mathbf{N} \in \mathbb{R}^p$ represents the transverse displacement upon hitting the planar boundary at a vertical distance $\lambda$. In \cite{lee24CFAP}, the authors established the following statistical properties for the vertically-drifted FAP noise:

\subsubsection{Probability Density Function}
The exact PDF of the noise vector, referred to as $\mathbf{N} \sim \mathrm{VDFAP}^{(p)}(u, \lambda)$ in \cite{lee24CFAP}, is a radially symmetric function involving the modified Bessel function of the second kind $K_\nu(\cdot)$, namely
\begin{equation}
\label{eq:FAP_pdf}
f_{\mathbf{N}}^{(p)}(\mathbf{n}) = \frac{2\lambda |u|^{\frac{p+1}{2}}}{(2\pi)^{\frac{p+1}{2}}} e^{\lambda |u|} \frac{K_{\frac{p+1}{2}} \left( |u| \sqrt{\|\mathbf{n}\|^2 + \lambda^2} \right)}{\left( \sqrt{\|\mathbf{n}\|^2 + \lambda^2} \right)^{\frac{p+1}{2}}},~~\mathbf{n} \in \mathbb{R}^p.
\end{equation}
This density explicitly captures the semi-heavy tail behavior, which decays as $\exp(-u\|\mathbf{n}\|)$ for large $\|\mathbf{n}\|$.
We note that the VDFAP notation represents the specific parameterization used in prior work.
In Section \ref{sec:NDFHL_distribution}, we will re-parameterize this probability law using a parameter set $(d, \lambda, u)$ and formally define it as the $\mathrm{NDFHL}$ distribution family to facilitate the unified capacity analysis.

\subsubsection{Characteristic Function}
The CF $\Phi_{\mathbf{N}}(\boldsymbol{\omega}) \triangleq \mathbb{E}[e^{i \langle \boldsymbol{\omega}, \mathbf{N}\rangle}]$ admits a compact closed form for any dimension $p$:
\begin{equation}
\label{eq:FAP_cf}
\Phi_{\mathbf{N}}^{(p)}(\boldsymbol{\omega}) = \exp \left( -\lambda \left[ \sqrt{\|\boldsymbol{\omega}\|^2 + |u|^2} - |u| \right] \right),
\end{equation}
where $\langle\cdot,\cdot\rangle$ denotes the standard Euclidean inner product.
As established in \cite{lee24CFAP}, letting $u \to 0$ in \eqref{eq:FAP_cf} directly recovers the CF of the multivariate Cauchy distribution.

\subsubsection{Moments}
The VDFAP noise is zero-mean ($\mathbb{E}[\mathbf{N}] = \mathbf{0}$), and its covariance matrix and second moment are given by
\begin{equation}
\label{eq:FAP_moments}
\mathbf{K}_{\mathbf{N}} \triangleq\mathbb{E}[\mathbf{N}\mathbf{N}^T]= \frac{\lambda}{|u|} \mathbf{I}_p, \quad \mathbb{E}[\|\mathbf{N}\|^2] = \frac{\lambda p}{|u|}.
\end{equation}
Note that the second moment diverges as $u \to 0$, revealing a physical singularity in the vanishing-drift regime.

\subsubsection{Differential Entropy (\texorpdfstring{$p=2$}{p=2} Special Case)}
While a general closed-form for the differential entropy $h(\mathbf{N})$ across all dimensions was not found yet, \cite{lee24CFAP} derived an explicit expression for the $p=2$ case (for 3D spatial system):
\begin{equation}
\label{eq:FAP_entropy_p2}
h(\mathbf{N}) = \log(2\pi e^3) + 2\log \lambda - \log(1+\lambda|u|) - \lambda|u| e^{\lambda|u|} \mathcal{I}(\lambda|u|),
\end{equation}
where $\mathcal{I}(s) = e \cdot \mathrm{Ei}(-1-s) - 3 \cdot \mathrm{Ei}(-s)$ and $\mathrm{Ei}(\cdot)$ is the exponential integral function. For general $p$, whether $h(\mathbf{N})$ remains finite and how it changes with $u$ is still unclear.

Notice that unless otherwise specified, all logarithms are natural.

\subsection{Theoretical Gaps and Motivation}

Despite these analytical results, \cite{lee24CFAP} left two fundamental questions unresolved, primarily due to the limitations of the complicated PDF forms.
\begin{itemize}
    \item \textbf{The Scaling Gap:} The capacity bounds established in \cite[Theorem~1, Theorem~2]{lee24CFAP} do not explicitly characterize the high-SNR scaling law. This is because the lower bound in \cite{lee24CFAP} is expressed in terms of the differential entropy of VDFAP noise, for which no general closed-form expression exists. While numerical simulations in \cite{lee24CFAP} for the $d=2$ case suggest that the upper and lower bounds converge as the power $P$ increases—indicating a $\log P$ scaling—a rigorous analytical proof remained elusive.
    
    \item \textbf{The Singularity at Zero-Drift Limit:} As shown in \eqref{eq:FAP_moments}, the noise power diverges as $u \to 0$. Consequently, the Max-Ent upper bound established in \cite{lee24CFAP} explodes (as shown in \cite[Fig.~5]{lee24CFAP}), leaving the capacity scaling of the Cauchy-limit regime analytically unresolved.
\end{itemize}

As a quick overview, the present work addresses these gaps by shifting the focus from the Bessel-type PDF to the latent GVM representation. To ensure analytical clarity, we adopt a refined notation: we denote the physical longitudinal drift by $\mu$ and introduce the \emph{normalized drift speed} $u = \mu/\sigma^2$. This allows us to provide a unified scaling law $C(P) = \frac{p}{2}\log P + \mathcal{L} + o(1)$ that remains robust even at the singular limit $u \to 0$. Here, $\mathcal{L}$ denotes the asymptotic capacity offset 
which will be discussed later in Section~\ref{sec:singular_limit}.


\section{New Observation: A Latent Gaussian Structure Underlying the First-Hitting Location}
\label{sec:stochastic_foundation}

To resolve the analytical challenges identified in Section~\ref{sec:previous_work}, we shift the focus from the intractable PDF to the underlying stochastic generation mechanism.
The core insight is the identification of the boundary-induced noise as a GVM.
This latent Gaussian structure not only simplifies the mathematical treatment but also provides the analytical foundation for our unified high-SNR capacity scaling analysis.

\begin{figure}[!ht]
    \centering
    \resizebox{\linewidth}{!}{%
    \begin{tikzpicture}[
        x={(1.4cm,0cm)},
        y={(0.5cm,0.4cm)},  
        z={(0cm,1cm)},      
        scale=2.0,
        >=stealth
    ]

    \def\hitX{0.7}    
    \def\hitY{0.6}    
    \def\lambdaDist{0.8} 
    \def\mutx{0.5}    
    \def\muty{0.2}    
    \def\mun{0.4}     

    \fill[blue!6, opacity=0.7] (-1.3,-1.2,0) -- (1.5,-1.2,0) -- (1.5,1.0,0) -- (-1.3,1.0,0) -- cycle;
    \draw[blue!35, dashed] (-1.3,-1.2,0) -- (1.5,-1.2,0) -- (1.5,1.0,0) -- (-1.3,1.0,0) -- cycle;

    \node[anchor=north west, font=\footnotesize, color=black!80, yshift=-2pt] at (-1.3,-1.2,0)
        {\textbf{Boundary plane} $\mathcal{B}:\ x_1 = \lambda$};

    \node[anchor=south east, font=\scriptsize, color=black!70] at (1.4,-1.1,0)
        {$\mathbf{N} \in \mathbb{R}^{d-1}$};

    \coordinate (O_icon) at (-1.1,-0.9,0.4); 
    \draw[->, thin, gray!60] (O_icon) -- ++(0,0,-0.3) node[below, font=\tiny] {$x_1$};

    \coordinate (X0) at (0,0,\lambdaDist);
    \fill[black] (X0) circle (1.6pt);
    \node[anchor=east, inner sep=6pt, font=\footnotesize] at (X0) {Initial point $\mathbf{0}$};

    \draw[dotted, thin, gray] (0,0,0) -- (X0);
    \draw[<->, thin] (0.15,0,0.04) -- (0.15,0,\lambdaDist+0.02)
        node[midway, right, font=\tiny, xshift=-1pt] {$\lambda$};

    \draw[->, thick, red!80!black] (X0) -- ++(\mutx,\muty,\mun)
        node[anchor=south west, inner sep=1pt, font=\footnotesize] {Uniform drift $\boldsymbol{\mu}$};

    \draw[->, thick, red!60!black] (X0) -- ++(\mutx,\muty,0)
        node[pos=1.0, anchor=north, font=\scriptsize, yshift=-2pt] {$\boldsymbol{\mu}_{\text{tan}}$}; 
    \draw[->, thick, red!60!black] (X0) -- ++(0,0,\mun)
        node[anchor=south, font=\scriptsize, yshift=1pt] {$\mu \mathbf{e}_1$};

    \coordinate (Hit) at (\hitX,\hitY,0);
    \draw[blue!80, thick, decorate,
          decoration={random steps, segment length=3pt, amplitude=2pt}]
        (X0) -- (0.25,0.15,0.6) -- (0.45,0.3,0.3) -- (Hit);

    \fill[red] (Hit) circle (1.6pt);

    \draw[->, >=latex, blue!40, thick] (0,0,0) -- (Hit) 
        node[midway, above, sloped, font=\tiny, inner sep=1pt] {$\mathbf{N}$};

    \node[anchor=west, align=left, font=\footnotesize, xshift=10pt] at (Hit) {
        \textbf{FHL} $\mathbf{N} \in \mathbb{R}^{d-1}$ \\
        at time $T$ \\
        {\scriptsize $X_T=(\mathbf{N}, \lambda)$}
    };

    \end{tikzpicture}%
    }
    \caption{Physical origin of the boundary-induced additive noise
    $\mathrm{NDFHL}^{(d)}$.
    A particle released from the origin $\mathbf{0}$ in the half-space
    $x_1<\lambda$ undergoes drift--diffusion and is absorbed upon first hitting the
    hyperplane $\mathcal{B}: x_1=\lambda$.
    The receiver records the first-hitting location
    $\mathbf{N}\in\mathbb{R}^{d-1}$ at the random hitting time $T$.
    This first-hitting location serves as the additive noise term in the equivalent
    channel model, with its statistics determined solely by the boundary geometry and the
    underlying drift--diffusion dynamics.}
\end{figure}
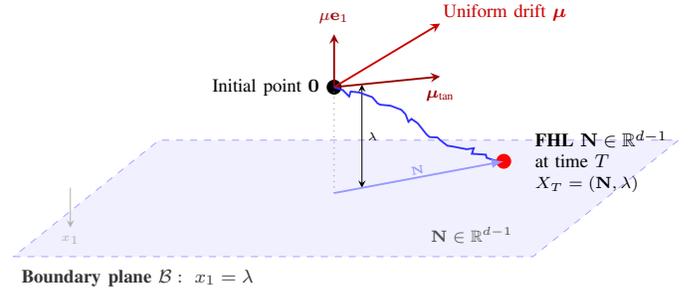

\subsection{System Model in \texorpdfstring{$\mathbb{R}^d$}{Rd}}

Consider an information carrier (particle) moving in a $d$-dimensional Euclidean space $\mathbb{R}^d$ ($d \ge 2$). Its position at time $t$ is decomposed into a longitudinal component $X_t \in \mathbb{R}$ and a $(d-1)$-dimensional transverse vector $\mathbf{Y}_t \in \mathbb{R}^{d-1}$. The particle initiates its motion from the origin at $t=0$, such that $X_0 = 0$ and $\mathbf{Y}_0 = \mathbf{0}$.

The motion of the particle is governed by the following system of stochastic differential equations (SDEs) \cite{oksendal2013stochastic}:
\begin{subequations}
\begin{align}
\label{eq:sde_system}
dX_t &= \mu dt + \sigma dW^{(x)}_t, \\
d\mathbf{Y}_t &= \sigma d\mathbf{W}^{(y)}_t,
\label{eq:sde_system2}
\end{align}
\end{subequations}
where $\mu$ represents the drift velocity along the longitudinal axis, and $\sigma>0$ is the diffusion coefficient. The terms $W^{(x)}_t$ and $\mathbf{W}^{(y)}_t$ denote a standard one-dimensional Brownian motion and a $(d-1)$-dimensional standard Brownian motion, respectively. Specifically, these processes are characterized by independent Gaussian increments, such that at any time $t > 0$:
\begin{equation}
W^{(x)}_t \sim \mathcal{N}(0, t), \quad \mathbf{W}^{(y)}_t \sim \mathcal{N}(\mathbf{0}, t\mathbf{I}_{d-1}),
\end{equation}
where $\mathbf{I}_{d-1}$ is the identity matrix of dimension $d-1$. We assume that all Brownian motion components are mutually independent.

Throughout this paper, the drift vector is assumed to be strictly orthogonal to the receiving boundary, with zero transverse components (i.e., $\boldsymbol{\mu}_{\text{tan}} = \mathbf{0}$). This configuration, termed \emph{Normally-Drifted}, is well-motivated by the MC literature. In diffusion-based MC systems, the longitudinal drift is primarily intended to assist transmission by directing the information carrier toward the receiver \cite{srinivas2012molecular}. Furthermore, many system-level optimizations in MC—such as those discussed in \cite{Ma2019,Cho2022,Chou2022}—are predicated on the assumption that this normal drift can be actively controlled or modulated to enhance performance. While the inclusion of tangential drift—a regime we term Tangentially-Drifted First-Hitting Location (TDFHL)—constitutes a meaningful extension, it is deferred to future investigation beyond the scope of this work.

To ensure a robust analysis, we consider the following parameter assumptions:
\begin{itemize}
    \item \textbf{Directionality:} Although the regime $\mu < 0$ is addressed in Section~\ref{subsec:reverse_drift}, it yields spatial statistics identical to the case with positive drift. Therefore, we assume $\mu > 0$ without loss of generality, focusing on particles directed toward the boundary.
    \item \textbf{Isotropy:} The diffusion coefficient $\sigma$ is assumed to be isotropic, affecting both longitudinal and transverse components equally (see \eqref{eq:sde_system} and \eqref{eq:sde_system2}).
\end{itemize}

\subsection{Temporal Statistics: The First-Passage Time}

The receiver is modeled as an infinite hyperplane located at a distance $\lambda > 0$ along the longitudinal axis, defined by the set $\mathcal{B} = \{ (x_1, \mathbf{y}) \in \mathbb{R} \times \mathbb{R}^{d-1} : x_1 = \lambda \}$. The random variable of  interest is the \textit{first-passage time} (or first-hitting time) $T$, defined as the first instant the particle contacts the boundary $\mathcal{B}$:
\begin{equation}
T = \inf \{ t > 0 : X_t = \lambda \}.
\end{equation}

For a Brownian motion with positive drift $\mu > 0$, the first-passage time $T$ is almost surely finite. Its probability density function follows the IG distribution \cite{srinivas2012molecular}, denoted as $T \sim IG(\nu, \kappa)$, with PDF
\begin{equation}
f_T(t) = \sqrt{\frac{\kappa}{2\pi t^3}} \exp\left( -\frac{\kappa (t - \nu)^2}{2\nu^2 t} \right), \quad t > 0,
\end{equation}
where the mean $\nu$ and the shape parameter $\kappa$ are determined by the physical parameters
\begin{equation}
\label{eq:IG_params}
\nu = \frac{\lambda}{\mu}, \quad \kappa = \frac{\lambda^2}{\sigma^2}.
\end{equation}

\begin{remark}[Parameter convention]
Throughout this paper, the physical drift velocity is denoted by $\mu$ and the
diffusion coefficient by $\sigma^2$.
We define the normalized drift speed as $u\triangleq\mu/\sigma^2$.
Accordingly, the first-hitting time satisfies
$\mathbb E[T]=\lambda/\mu=\lambda/(u\sigma^2)$.
All subsequent expressions are written consistently under this convention.
\end{remark}

\subsection{Spatial Dynamics: Boundary-Induced Noise}

Upon arriving at the boundary $x = \lambda$ at the random time $T$, the particle's location on the hyperplane is determined by the state of the transverse process $\mathbf{Y}_t$ at that instant. We define the \textit{first-hitting location} random vector $\mathbf{N} \in \mathbb{R}^{d-1}$ as
\begin{equation}
\mathbf{N} = \mathbf{Y}_T.
\end{equation}
Since the transverse process $\mathbf{Y}_t$ is a drift-free Brownian motion independent of the longitudinal process $X_t$ (and thus independent of the stopping time $T$), the conditional distribution of $\mathbf{N}$ given the arrival time $T=t$ is a multivariate Gaussian distribution (please compare to \eqref{eq:sde_system2}):
\begin{equation}
\mathbf{N} \mid T=t \sim \mathcal{N}(\mathbf{0}, \sigma^2 t \mathbf{I}_{d-1}).
\end{equation}
Here, $\mathbf{I}_{d-1}$ denotes the identity matrix of dimension $d-1$.

This construction identifies the marginal distribution of $\mathbf{N}$ as a GVM, with the first-passage time $T$ serving as the IG mixing variable. Specifically, $\mathbf{N}$ admits the following stochastic representation:
\begin{equation}
\label{eq:GVM_rep}
\mathbf{N} \stackrel{d}{=} \sigma \sqrt{T} \mathbf{Z},
\end{equation}
where $\mathbf{Z} \sim \mathcal{N}(\mathbf{0}, \mathbf{I}_{d-1})$ is a standard normal random vector independent of $T$. In \eqref{eq:GVM_rep}, the notation $\stackrel{d}{=}$ denotes equality in distribution.

\subsection{Physical Interpretation and Statistical Context}Beyond its mathematical definition, the normalized drift speed $u$ (with units $L^{-1}$) serves as a critical transport regime indicator. A large $u$ characterizes a ``ballistic'' regime where the particle acts like a directed bullet toward the boundary, whereas the limit $u \to 0$ represents a purely ``diffusive'' regime resembling the random walk of pollen. This fundamental spatial uncertainty as $u \to 0$ provides the physical origin for the heavy-tailed behavior and the resulting singularity in noise variance observed in boundary-induced processes.

\begin{remark}[On related distributions and novelty] While \eqref{eq:GVM_rep} identifies the NDFHL noise as a member of the Normal-Inverse-Gaussian (NIG) subclass within the Generalized Hyperbolic (GH) family, our novelty lies in the identification of its \emph{geometric generation mechanism}. Unlike standard statistical studies that treat shape parameters as abstract fitting variables, our framework derives the noise structure directly from drift--diffusion transport dynamics. This approach allows us to analyze \textit{capacity scaling laws} through the lens of transport geometry—a direction previously unexplored in the NIG/GH literature. \end{remark}


\section{The \texorpdfstring{$\mathrm{NDFHL}^{(d)}$}{NDFHL(d)} Distribution Family}
\label{sec:NDFHL_distribution}

The stochastic representation established in Section \ref{sec:stochastic_foundation} reveals that the boundary-induced hitting location noise arises from a GVM modulated by an IG time scale. 
By adopting the normalized drift $u$ as the governing parameter, this physical process crystallizes into a distinct and mathematically tractable family of distributions. 
In this section, we formally define this family as the \emph{Normally-Drifted First-Hitting Location} (NDFHL) distribution. 
We derive its generalized PDF, prove its radial symmetry, and analyze its moments, providing the rigorous mathematical footing required to resolve the capacity scaling gaps identified in Section \ref{sec:previous_work}.

\subsection{Definition and Probability Density Function}

Let the ambient space dimension be $d \geq 2$. We consider a random vector $\mathbf{N}$ taking values in a $(d-1)$-dimensional hyperplane $\mathbb{R}^{d-1}$. The distribution of $\mathbf{N}$ is parameterized by the boundary distance $\lambda > 0$ and the normalized drift speed $u \in \mathbb{R}$. For the case $u < 0$, the $\mathrm{NDFHL}$ distribution becomes a so-called \emph{defective distribution} with total probability less than $1$; however, due to the symmetry property established in Section~\ref{subsec:reverse_drift}, the spatial shape of this defective distribution (conditioned on hitting) is identical to the $\mathrm{NDFHL}$ with $|u| > 0$. As a result, without loss of generality, we will assume that $u > 0$ throughout this section. The case $u\to 0$ will be discussed in Section~\ref{sec:singular_limit}.

\begin{definition}[$\mathrm{NDFHL}^{(d)}$ Distribution]
A random vector $\mathbf{N} \in \mathbb{R}^{d-1}$ is said to follow the \emph{Normally-Drifted First-Hitting Location} distribution with parameters $(d, \lambda, u)$, denoted as $\mathbf{N} \sim \mathrm{NDFHL}^{(d)}(\lambda, u)$, if its probability density function is given by:
\begin{equation}
\label{eq:NDFHL_pdf}
f_{\mathbf{N}}(\mathbf{n}) = \frac{\lambda}{2^{\frac{d}{2}-1} \pi^{\frac{d}{2}}} \left( \frac{u}{\rho(\mathbf{n})} \right)^{\frac{d}{2}} e^{\lambda u} K_{\frac{d}{2}}\left( u \rho(\mathbf{n}) \right),
\end{equation}
where $\mathbf{n} \in \mathbb{R}^{d-1}$, and $\rho(\mathbf{n}) \triangleq \sqrt{\|\mathbf{n}\|^2 + \lambda^2}$ represents the Euclidean distance from the origin of the ambient space to the location $\mathbf{n}$ on the hyperplane. Here, $K_{\nu}(\cdot)$ denotes the modified Bessel function of the second kind of order $\nu$.
\end{definition}

Equation \eqref{eq:NDFHL_pdf} can be viewed as a refined re-parameterization of the VDFAP law introduced in prior work, specifically tailored to facilitate our geometric scaling analysis. The normalization of \eqref{eq:NDFHL_pdf} is consistency established by the characteristic function (Lemma~\ref{lem:char_func}) and the latent IG–Gaussian mixture representation in \eqref{eq:GVM_rep}.

For notational convenience, we denote the effective dimension of the noise vector (and the receiving boundary) by: 
\begin{equation}
p \triangleq d-1. 
\end{equation}
Throughout the remainder of this paper, $\mathbb{R}^p$ and $\mathbb{R}^{d-1}$ are used interchangeably to represent the transverse spatial domain.


The order of the Bessel function, $\nu = d/2$, reflects the dimension of the underlying diffusion process. The term $\rho(\mathbf{n})$ encapsulates the geometry of the first-hitting problem, coupling the radial displacement $\|\mathbf{n}\|$ with the propagation distance $\lambda$.

\subsection{Low Dimensional Special Cases}
\label{subsec:special_cases}

To provide physical intuition, we examine the PDF \eqref{eq:NDFHL_pdf} in the two most relevant dimensions for physical modeling and for MC context.

\subsubsection{The Planar Case (\texorpdfstring{$d=2$}{d=2})}
Consider a particle moving in a two-dimensional plane and hitting a one-dimensional line boundary ($p=1$). Setting $d=2$, the Bessel order becomes $\nu=1$, and the noise reduces to a scalar random variable $N \in \mathbb{R}$. The PDF simplifies to:
\begin{equation}
f_{N}(n) = \frac{\lambda u}{\pi \sqrt{n^2 + \lambda^2}} e^{\lambda u} K_1\left( u \sqrt{n^2 + \lambda^2} \right).
\end{equation}
This form is recognized in the molecular communication literature as the hitting probability density for a point receiver in a 2D diffusive environment \cite{pandey2018molecular}.

\subsubsection{The Spatial Case (\texorpdfstring{$d=3$}{d=3})}
Consider a particle moving in three-dimensional space hitting a two-dimensional planar boundary ($p=2$). Setting $d=3$, the noise is a vector $\mathbf{N} \in \mathbb{R}^2$. The PDF is first expressed by directly substituting the dimension into \eqref{eq:NDFHL_pdf}, yielding the Bessel form:
\begin{equation}
\label{eq:NDFHL_d3_bessel}
f_{\mathbf{N}}(\mathbf{n}) = \frac{\lambda}{\sqrt{2} \pi^{3/2}} \left( \frac{u}{\rho(\mathbf{n})} \right)^{3/2} e^{\lambda u} K_{3/2}\left( u \rho(\mathbf{n}) \right).
\end{equation}
This representation highlights the distribution's membership in the $\mathrm{NDFHL}$ family with a half-integer index $\nu = 3/2$. Using the closed-form identity for half-integer modified Bessel functions, $K_{3/2}(z) = \sqrt{\frac{\pi}{2z}} e^{-z} (1 + \frac{1}{z})$, the PDF reduces to an elementary function form \cite{lee2016distribution}:
\begin{equation}
\label{eq:NDFHL_d3_elementary}
f_{\mathbf{N}}(\mathbf{n}) = \frac{\lambda}{2\pi \rho(\mathbf{n})^3} (1 + u \rho(\mathbf{n})) \exp\left( -u (\rho(\mathbf{n}) - \lambda) \right).
\end{equation}
The comparison between \eqref{eq:NDFHL_d3_bessel} and \eqref{eq:NDFHL_d3_elementary} reveals the underlying geometric structure: the term $\lambda / (2\pi \rho^3)$ is precisely the Poisson kernel for the half-space, representing the hitting density of a pure diffusion process ($u=0$), while the remaining terms $(1 + u\rho) e^{-u(\rho-\lambda)}$ characterize the spatial confinement and exponential tail decay induced by the normalized drift $u$.

\subsection{Radial Symmetry and Geometric Properties}

The density function \eqref{eq:NDFHL_pdf} exhibits explicit radial symmetry. Since $f_{\mathbf{N}}(\mathbf{n})$ depends on the vector $\mathbf{n}$ solely through its norm $\|\mathbf{n}\|$, the distribution is isotropic on the hyperplane $\mathbb{R}^{d-1}$.

\begin{proposition}[Isotropy]
Let $\mathbf{N} \sim \mathrm{NDFHL}^{(d)}(\lambda, u)$. For any orthogonal matrix $\mathbf{Q} \in O(d-1)$, the random vector $\mathbf{Q}\mathbf{N}$ has the same distribution as $\mathbf{N}$.
\end{proposition}

This symmetry implies that the noise induced by the boundary is unbiased in the transverse directions. The ``semi-heavy-tailed'' nature of the distribution arises from the interaction between the exponential term $e^{\lambda u}$ and the asymptotic decay of the Bessel function. Using the asymptotic expansion $K_{\nu}(z) \approx \sqrt{\frac{\pi}{2z}} e^{-z}$ for large $z$, the PDF behaves as
\begin{equation}
f_{\mathbf{N}}(\mathbf{n}) \sim C \cdot \|\mathbf{n}\|^{-\frac{d+1}{2}} \exp\left( -u (\|\mathbf{n}\| - \lambda) \right) \quad \text{as } \|\mathbf{n}\| \to \infty.
\label{eq:semi-HT}
\end{equation}
Eq.~\eqref{eq:semi-HT} indicates that while the drift induces an exponential cutoff, the probability of large deviations remains significantly higher than that of a Gaussian distribution with equivalent variance, particularly for small $u$.

\subsection{Statistical Moments}

Due to the spherical symmetry derived above, the first moment (i.e., mean) of the distribution is centered at the origin.

\begin{proposition}[Mean]\label{prop:mean_zero}
If $\mathbf{N} \sim \mathrm{NDFHL}^{(d)}(\lambda, u)$, then the expectation exists and is given by:
\begin{equation}
\mathbb{E}[\mathbf{N}] = \mathbf{0}.
\end{equation}
\end{proposition}

Due to the isotropy of the hitting process, the covariance matrix of $\mathbf{N}$ is proportional to the identity matrix, i.e., $\mathrm{Cov}(\mathbf{N}) = \sigma_N^2 \mathbf{I}_{d-1}$. We determine the component-wise variance $\sigma_N^2$ by leveraging the GVM structure established in Section~\ref{sec:stochastic_foundation}. Recall that $\mathbf{N}$ admits the stochastic representation $\mathbf{N} \stackrel{d}{=} \sigma \sqrt{T} \mathbf{Z}$ , where $\mathbf{Z} \sim \mathcal{N}(\mathbf{0}, \mathbf{I}_{d-1})$ and $T \sim IG(\lambda/\mu, \lambda^2/\sigma^2)$.

\begin{lemma}[Covariance]
The covariance matrix of $\mathbf{N} \sim \mathrm{NDFHL}^{(d)}(\lambda, u)$ is given by:
\begin{equation}
\mathbf{K}_{\mathbf{N}}=\mathbb{E}[\mathbf{N}\mathbf{N}^T] = \frac{\lambda}{u} \mathbf{I}_{d-1}.
\end{equation}
\label{lem:variance}
\end{lemma}

\begin{IEEEproof}
By the law of total variance and the independence of $T$ and $\mathbf{Z}$:
\begin{align}
\text{Var}(N_i) &= \mathbb{E}[\text{Var}(N_i | T)] + \text{Var}(\mathbb{E}[N_i | T]) \nonumber \\
&= \mathbb{E}[\sigma^2 T] + \text{Var}(0) \nonumber \\
&= \sigma^2 \mathbb{E}[T].
\end{align}
Substituting the mean of the Inverse Gaussian distribution, $\mathbb{E}[T] = \frac{\lambda}{\mu}$, and using the definition of normalized drift $u = \mu/\sigma^2$, we obtain:
\begin{equation}
\text{Var}(N_i) = \sigma^2 \left( \frac{\lambda}{\mu} \right) = \frac{\lambda}{\mu/\sigma^2} = \frac{\lambda}{u}.
\end{equation}

The vanishing of the off-diagonal entries in $\mathbf{K}_{\mathbf{N}}$ follows directly from the independence of the spatial components of the underlying Brownian motion. Specifically, since the transverse process $\mathbf{Y}_t = (Y^{(1)}_t, \dots, Y^{(d-1)}_t)$ consists of $d-1$ mutually independent and drift-free Brownian motions, we have $\text{Cov}(Y^{(i)}_t, Y^{(j)}_t) = 0$ for all $t > 0$ whenever $i \neq j$. Furthermore, because the transverse process $\mathbf{Y}_t$ is independent of the longitudinal process $X_t$, the hitting time $T$ (which is determined solely by $X_t$) remains independent of the spatial fluctuations of $\mathbf{Y}_t$. By the law of total covariance, we obtain:
\begin{equation}
\label{eq:cov_derivation}
\begin{aligned}
\operatorname{Cov}(N_i, N_j) &= \mathbb{E}[\operatorname{Cov}(Y^{(i)}_T, Y^{(j)}_T \mid T)] \\
&\quad + \operatorname{Cov}(\mathbb{E}[Y^{(i)}_T \mid T], \mathbb{E}[Y^{(j)}_T \mid T]) \\
&= 0 + 0 = 0, \quad i \neq j.
\end{aligned}
\end{equation}
This confirms that the boundary-induced noise inherits the spatial isotropy of the diffusion process.

\end{IEEEproof}


\subsection{Two Limiting Cases}

The $\mathrm{NDFHL}^{(d)}$ family connects two distinct physical regimes through the parameter $u$.

\subsubsection{Drift-Dominated Regime (\texorpdfstring{$u \to \infty$}{u -> inf})}
In the drift-dominated regime, the hitting location $\mathbf{N}$ becomes increasingly concentrated around the longitudinal projection of the source. As $u \to \infty$, the strong drift suppresses the heavy-tailed fluctuations of the first-passage time $T$, causing the spatial distribution of $\mathbf{N}$ to converge toward a multivariate Gaussian law with a vanishing covariance matrix $\frac{\lambda}{u} \mathbf{I}_{d-1}$. This characterizes the ``ballistic'' limit, where the information carrier acts as a directed projectile toward the boundary, effectively minimizing the transverse dispersion induced by diffusion.

\subsubsection{Diffusion-Dominated Regime (\texorpdfstring{$u \to 0$}{u -> 0})}
In the limit of vanishing drift, the distribution approaches a generalized Cauchy distribution \cite{Verdu_Entropy23}. Specifically, for $u \to 0$, $K_{\nu}(z) \approx \frac{1}{2} \Gamma(\nu) (z/2)^{-\nu}$. Substituting this into \eqref{eq:NDFHL_pdf}:
\begin{equation}
\lim_{u \to 0} f_{\mathbf{N}}(\mathbf{n}) \propto \frac{\lambda}{\rho(\mathbf{n})^d} = \frac{\lambda}{(\|\mathbf{n}\|^2 + \lambda^2)^{d/2}}.
\end{equation}
This recovers the classical Poisson kernel for the half-space in $d$ dimensions, confirming that purely diffusive hitting probabilities are heavy-tailed (Cauchy-like) with undefined variance.

\subsection{Reverse Drift: The Defective \texorpdfstring{$\mathrm{NDFHL}$}{NDFHL} Family}
\label{subsec:reverse_drift}

Recall that the $\mathrm{NDFHL}$ family defined in \eqref{eq:NDFHL_pdf} assumes a positive normalized drift $u > 0$, ensuring that the particle hits the boundary with probability one. We now consider the reverse drift scenario, $u < 0$, where the drift pulls the particle away from the boundary.

In this regime, the first-passage time $T$ is no longer a proper random variable; there is a non-zero probability that the particle wanders off to infinity without ever crossing the hyperplane $x_1 = \lambda$. The total \textit{hitting probability} (i.e., survival probability) is given by the well-known exponential decay \cite{redner2001guide,MortersPeres:2010}:
\begin{equation}
P_{\text{hit}} \triangleq \mathbb{P}(T < \infty) = e^{-2\lambda |u|}.
\end{equation}
The complement, $1 - P_{\text{hit}}$, represents the vanishing probability.

Consequently, the distribution of the hitting location $\mathbf{N}$ becomes a \emph{defective distribution}. The unnormalized density, $f_{\mathbf{N}}^{\text{def}}(\mathbf{n})$, describes the likelihood of hitting the boundary at location $\mathbf{n}$. This defective density satisfies a remarkable symmetry property:
\begin{equation}
\label{eq:defective_pdf}
f_{\mathbf{N}}^{\text{def}}(\mathbf{n}; u) = e^{-2\lambda |u|} \cdot f_{\mathbf{N}}(\mathbf{n}; |u|),
\end{equation}
where $f_{\mathbf{N}}(\mathbf{n}; |u|)$ is the standard, normalized $\mathrm{NDFHL}^{(d)}$ PDF defined in \eqref{eq:NDFHL_pdf} with positive drift parameter $|u|$.

Substituting the definition of the $\mathrm{NDFHL}$ density, the explicit form of the defective density reads
\begin{equation}
f_{\mathbf{N}}^{\text{def}}(\mathbf{n}; u) = \frac{\lambda}{2^{\frac{d}{2}-1} \pi^{\frac{d}{2}}} \left( \frac{|u|}{\rho(\mathbf{n})} \right)^{\frac{d}{2}} e^{-\lambda |u|} K_{\frac{d}{2}}\left( |u| \rho(\mathbf{n}) \right).
\end{equation}
Note that the exponential term here is $e^{-\lambda |u|}$ (derived from $e^{-2\lambda|u|} \times e^{\lambda|u|}$), which ensures the integral over $\mathbb{R}^{d-1}$ equals $P_{\text{hit}} < 1$.

Equation \eqref{eq:defective_pdf} offers a profound physical interpretation: \emph{conditional on the event that the particle hits the boundary}, the spatial statistics of a particle fighting against a drift $-|u|$ are identical to those of a particle aided by a drift $+|u|$. The reverse drift suppresses the \emph{rate} of arrival, but it does not alter the \emph{shape} of the spatial dispersion for those trajectories that successfully arrive.

By integrating the defective NDFHL regime with the standard family, the NDFHL can be viewed as a generalized distribution class defined for all $u \in \mathbb{R}$. From this perspective, the NDFHL framework effectively generalizes the original VDFAP family proposed in prior work, providing a unified characterization that encompasses both certain and defective transport processes.


\section{Further Analytical Properties of the \texorpdfstring{$\mathrm{NDFHL}^{(d)}$}{NDFHL(d)} Family}
\label{sec:properties}

Having formally defined the $\mathrm{NDFHL}^{(d)}$ distribution and established its basic statistical properties, we now turn to its frequency-domain characterization, algebraic structure, and some entropy properties. 



\subsection{Characteristic Function}

The CF is a powerful tool for analyzing sums of independent random variables.
\begin{lemma}[Characteristic Function of $\mathrm{NDFHL}$]
\label{lem:char_func}
Let $\mathbf{N} \sim \mathrm{NDFHL}^{(d)}(\lambda, u)$ with $\lambda > 0$ and $u > 0$. The CF $\Phi_\mathbf{N}(\boldsymbol{\omega}) \triangleq \mathbb{E}[e^{i\langle \boldsymbol{\omega}, \mathbf{N} \rangle}]$ for $\omega \in \mathbb{R}^{d-1}$ is given by:
\begin{equation}
    \Phi_\mathbf{N}(\boldsymbol{\omega}) = \exp\left( -\lambda \left[ \sqrt{u^2 + \|\boldsymbol{\omega}\|^2} - u \right] \right).
    \label{eq:cf_closed_form}
\end{equation}
\end{lemma}

\begin{IEEEproof}
We exploit the stochastic representation $\mathbf{N} = \sigma\sqrt{T}\mathbf{Z}$, where $\mathbf{Z} \sim \mathcal{N}(0, \mathbf{I}_{d-1})$ and $T \sim IG(\lambda/\mu, \lambda^2/\sigma^2)$. By the law of iterated expectations:
\begin{align}
    \Phi_\mathbf{N}(\boldsymbol{\omega}) &= \mathbb{E}_T\left[ \mathbb{E}_\mathbf{Z}\left[ e^{i\langle \boldsymbol{\omega}, \sigma\sqrt{T}\mathbf{Z} \rangle} \mid T \right] \right] \nonumber \\
    &= \mathbb{E}_T\left[ \exp\left( -\frac{1}{2} \|\sigma\sqrt{T}\boldsymbol{\omega}\|^2 \right) \right] \nonumber \\
    &= \mathbb{E}_T\left[ \exp\left( -\frac{\sigma^2 \|\boldsymbol{\omega}\|^2}{2} T \right) \right].
\label{eq:CF_step1}
\end{align}
The expectation in \eqref{eq:CF_step1} corresponds to the Moment Generating Function (MGF) of the Inverse Gaussian variable $T$, evaluated at $s = - \frac{\sigma^2 \|\boldsymbol{\omega}\|^2}{2}$.
The MGF of $T \sim IG(\nu, \kappa)$ is known as \cite{Chhikara:1989}: 
\begin{equation}
M_T(s) = \exp\left( \frac{\kappa}{\nu} \left[ 1 - \sqrt{1 - \frac{2\nu^2 s}{\kappa}} \right] \right),
\end{equation}
for $s \leq \frac{\kappa}{2\nu^2}$. 
Substituting the physical parameters $\nu = \lambda/\mu$ and $\kappa = \lambda^2/\sigma^2$, and utilizing the normalized drift $u = \mu/\sigma^2$, we have:
\begin{itemize}
    \item Ratio $\frac{\kappa}{\nu} = \frac{\lambda^2/\sigma^2}{\lambda/(u\sigma^2)} = \lambda u$.
    \item Term inside the square root:
    \begin{align}
    1 - \frac{2\nu^2}{\kappa} s &= 1 - \frac{2(\lambda^2/u^2 \sigma^4)}{\lambda^2/\sigma^2} \left( -\frac{\sigma^2 \|\boldsymbol{\omega}\|^2}{2} \right) \nonumber \\
    &= 1 + \frac{2}{u^2 \sigma^2} \cdot \frac{\sigma^2 \|\boldsymbol{\omega}\|^2}{2} \nonumber \\
    &= 1 + \frac{\|\boldsymbol{\omega}\|^2}{u^2} = \frac{u^2 + \|\boldsymbol{\omega}\|^2}{u^2}.
    \end{align}
\end{itemize}
Substituting these back into the MGF expression yields:
\begin{align}
\Phi_{\mathbf{N}}(\boldsymbol{\omega}) &= \exp\left( \lambda u \left[ 1 - \sqrt{\frac{u^2 + \|\boldsymbol{\omega}\|^2}{u^2}} \right] \right) \nonumber \\
&= \exp\left( \lambda u - \lambda \sqrt{u^2 + \|\boldsymbol{\omega}\|^2} \right),
\end{align}
which rearranges to \eqref{eq:CF_step1}.
\end{IEEEproof}

Lemma~\ref{lem:char_func} offers immediate insights. As $u \to 0$, the exponent approaches $-\lambda \|\boldsymbol{\omega}\|$, correctly recovering the CF of the multivariate Cauchy distribution. Conversely, as $u \to \infty$, a Taylor expansion shows $\Phi_{\mathbf{N}}(\boldsymbol{\omega}) \to \exp(-\frac{\lambda}{2u} \|\boldsymbol{\omega}\|^2)$, recovering the Gaussian limit.

To see this more clearly, we rewrite the exponent in \eqref{eq:CF_step1} as $\lambda u [ 1 - \sqrt{1 + (\|\boldsymbol{\omega}\|^2/u^2)} ]$. As $u \to \infty$, the term $\|\boldsymbol{\omega}\|^2/u^2$ becomes a small perturbation. By applying the first-order Taylor expansion $\sqrt{1+x} = 1 + \frac{x}{2} + O(x^2)$ for $x \to 0$, the exponent simplifies to
\begin{equation}
\lambda u \left[ 1 - \left( 1 + \frac{\norm{\boldsymbol{\omega}}^2}{2u^2} + \dots \right) \right] \approx -\frac{\lambda \norm{\boldsymbol{\omega}}^2}{2u}.
\end{equation}
The resulting CF $\Phi_{\mathbf{N}}(\boldsymbol{\omega}) \approx \exp(-\frac{\lambda}{2u} \norm{\boldsymbol{\omega}}^2)$ is the characteristic function of a multivariate Gaussian.

\subsection{Convolution Closure and Infinite Divisibility}

Another fundamental property of the NDFHL family is its closure under convolution with respect to the distance parameter $\lambda$. This property is a direct consequence of the semigroup property of the underlying Brownian motion.

\begin{lemma}[Convolution Closure]
\label{lem:convolution}
Let $\mathbf{N}_1 \sim \mathrm{NDFHL}^{(d)}(\lambda_1, u)$ and $\mathbf{N}_2 \sim \mathrm{NDFHL}^{(d)}(\lambda_2, u)$ be independent random vectors sharing the same dimension and normalized drift $u$. Then their sum follows:
\begin{equation}
\mathbf{N}_1 + \mathbf{N}_2 \sim \mathrm{NDFHL}^{(d)}(\lambda_1 + \lambda_2, u).
\end{equation}
\end{lemma}

\begin{IEEEproof}
Let $\Phi_1(\boldsymbol{\omega})$ and $\Phi_2(\boldsymbol{\omega})$ be the CFs of $\mathbf{N}_1$ and $\mathbf{N}_2$. Since they are independent, the CF of the sum is the product of their CFs:
\begin{align}
\Phi_{\mathbf{N}_1 + \mathbf{N}_2}(\boldsymbol{\omega}) &= \Phi_1(\boldsymbol{\omega}) \cdot \Phi_2(\boldsymbol{\omega}) \nonumber \\
&= e^{-\lambda_1 (\sqrt{u^2 + \|\boldsymbol{\omega}\|^2} - u)} \cdot e^{-\lambda_2 (\sqrt{u^2 + \|\boldsymbol{\omega}\|^2} - u)} \nonumber \\
&= e^{-(\lambda_1 + \lambda_2) (\sqrt{u^2 + \|\boldsymbol{\omega}\|^2} - u)}.
\end{align}
By the uniqueness theorem of CFs, this is the CF of an $\mathrm{NDFHL}$ variable with distance parameter $\lambda_1 + \lambda_2$.
\end{IEEEproof}

\begin{corollary}[Infinite Divisibility]
The $\mathrm{NDFHL}^{(d)}$ distribution is infinitely divisible. That is, for any integer $k \ge 1$, a random vector $\mathbf{N} \sim \mathrm{NDFHL}^{(d)}(\lambda, u)$ can be represented as the sum of $k$ independent and identically distributed random vectors $\mathbf{N}^{(1)}, \dots, \mathbf{N}^{(k)}$, where each $\mathbf{N}^{(j)} \sim \mathrm{NDFHL}^{(d)}(\lambda/k, u)$.
\label{cor:ID}
\end{corollary}

The infinite divisibility of $\mathbf{N}$ reflects the spatial homogeneity and Markovian nature of the underlying transport process. Specifically, the transverse displacement accumulated over a distance $\lambda$ can be decomposed into the vector sum $\mathbf{N} = \sum_{j=1}^k \mathbf{N}^{(j)}$, where each $\mathbf{N}^{(j)}$ represents the i.i.d. increment over a slab of thickness $\lambda/k$. Since this additive decomposition holds for any $k$, the NDFHL distribution is inherently infinitely divisible. This property provides a consistent analytical framework for multi-hop relay channels in MC context, where the total noise is the superposition of independent link-wise displacements.

\subsection{Entropy Structure and Bounds}
\label{subsec:entropy_analysis}

While the exact differential entropy $h(\mathbf{N})$ does not admit a simple elementary closed form for general dimensions, its structural behavior can be fully constrained by the mixture representation. By decomposing the entropy into a conditional component and a mutual information term, $h(\mathbf{N}) = h(\mathbf{N}|T) + I(\mathbf{N};T)$, we can sandwich the entropy between tractable physical bounds.

\begin{lemma}[Entropy Bounds and Decomposition]
\label{lem:entropy_bounds}
Fix any normalized drift $u > 0$, the differential entropy of the NDFHL noise vector $\mathbf{N} \sim \mathrm{NDFHL}^{(d)}(\lambda, u)$ is bounded by:
\begin{equation}
\label{eq:entropy_inequality}
\underline{h}(u) \le h(\mathbf{N}) \le \bar{h}(u),
\end{equation}
where the lower bound is the conditional entropy given the hitting time:
\begin{equation}
\label{eq:lower_bound_def}
\underline{h}(u) \triangleq h(\mathbf{N}|T) = \frac{p}{2} \log(2\pi e \sigma^2) + \frac{p}{2} \mathbb{E}[\log T],
\end{equation}
and the upper bound is the maximum entropy of a Gaussian vector with covariance $\frac{\lambda}{u}\mathbf{I}_p$:
\begin{equation}
\label{eq:upper_bound_def}
\bar{h}(u) \triangleq \frac{p}{2} \log\left( 2\pi e \frac{\lambda}{u} \right).
\end{equation}
Consequently, the entropy gap (or shape penalty) $I(\mathbf{N};T) \triangleq h(\mathbf{N}) - h(\mathbf{N}|T)$ is determined by the Jensen gap of the mixing time $T$:
\begin{equation}
\label{eq:jensen_gap}
0 \le h(\mathbf{N}) - \underline{h}(u) 
\le \bar{h}(u) - \underline{h}(u)
\le \frac{p}{2} \left( \log \mathbb{E}[T] - \mathbb{E}[\log T] \right).
\end{equation}
\end{lemma}

\begin{IEEEproof}

\textbf{1) Lower Bound (Conditional Entropy):} Recall the stochastic representation $\mathbf{N} \mid T=t \sim \mathcal{N}(\mathbf{0}, \sigma^2 t \mathbf{I}_p)$. 

Also recall that for a $p$-dimensional Gaussian random vector $\mathbf{X} \sim \mathcal{N}(\mathbf{0}, \Sigma)$ with covariance $\Sigma \succ 0$, the differential entropy is given by the closed-form expression
\begin{equation}
\label{eq:entropy_gaussian}
h(\mathbf{X}) = \frac{1}{2} \log\Big( (2\pi e)^p |\Sigma|\Big),
\end{equation}
where $|\cdot|$ denotes the determinant.
This entropy depends solely on the covariance matrix and is invariant under orthogonal transformations. In the isotropic case where $\Sigma = \alpha \mathbf{I}_p$, \eqref{eq:entropy_gaussian} simplifies to
\begin{equation}
h(\mathbf{X}) = \frac{p}{2} \log(2\pi e \alpha),\quad \alpha>0.
\end{equation}

Thus for our case, the entropy of this conditional Gaussian vector is:
\begin{equation}
h(\mathbf{N} | T=t) = \frac{p}{2} \log(2\pi e \sigma^2) + \frac{p}{2} \log t.
\end{equation}
Averaging over the distribution of $T$, we obtain the conditional entropy:
\begin{equation}
\begin{aligned}
h(\mathbf{N}|T) &= \int_0^\infty h(\mathbf{N}|T=t) f_T(t) \,\mathrm{d}t \\
&= \frac{p}{2} \log(2\pi e \sigma^2) + \frac{p}{2} \mathbb{E}[\log T].
\end{aligned}
\end{equation}
Since conditioning reduces entropy (i.e., $I(\mathbf{N};T) \ge 0$), it follows that $h(\mathbf{N}) \ge h(\mathbf{N}|T)$, establishing the lower bound $\underline{h}(u)$.

\textbf{2) Upper Bound (Max-Ent Principle):} From Lemma~\ref{lem:variance}, the covariance matrix of $\mathbf{N}$ is $\mathbf{K}_{\mathbf{N}} = \frac{\lambda}{u} \mathbf{I}_p$. Differential entropy is maximized by the Gaussian distribution for a fixed covariance:
\begin{equation}
\begin{aligned}
h(\mathbf{N}) &\le \frac{1}{2} \log \Big( (2\pi e)^p |\mathbf{K}_{\mathbf{N}}| \Big) \\
&= \frac{p}{2} \log\left( 2\pi e \frac{\lambda}{u} \right),
\end{aligned}
\end{equation}
which establishes the upper bound $\bar{h}(u)$.

\textbf{3) The Gap:} Combining \eqref{eq:lower_bound_def} and \eqref{eq:upper_bound_def}, the difference is:
\begin{equation}
\begin{aligned}
\bar{h}(u) - \underline{h}(u) &= \frac{p}{2} \left[ \log(\lambda/u) - \log(\sigma^2) - \mathbb{E}[\log T] \right] \\
&= \frac{p}{2} \left[ \log\left( \frac{\lambda}{u\sigma^2} \right) - \mathbb{E}[\log T] \right].
\end{aligned}
\end{equation}
Recalling $u = \mu/\sigma^2$, we have $\frac{\lambda}{u\sigma^2} = \frac{\lambda}{\mu} = \mathbb{E}[T]$. Thus:
\begin{equation}
\bar{h}(u) - \underline{h}(u) = \frac{p}{2} \left( \log \mathbb{E}[T] - \mathbb{E}[\log T] \right).
\end{equation}
By Jensen's inequality and the strict concavity of $\log(\cdot)$, $\log \mathbb{E}[T] \ge \mathbb{E}[\log T]$, confirming the gap is non-negative.
\end{IEEEproof}

\section{Capacity Scaling Law Under\ Second-Moment Constraint}\label{sec:scaling}Having established the statistical properties and the GVM structure of the NDFHL family in Sections \ref{sec:NDFHL_distribution} and \ref{sec:properties}, we are now positioned to quantify the ultimate information transfer capability supported by this transport mechanism. Specifically, we leverage the finite second moment (Lemma~\ref{lem:variance}) and the entropy bounds (Lemma~\ref{lem:entropy_bounds}) to derive the capacity scaling law of the boundary-induced additive noise channel
\begin{equation}\label{eq:V_channel_display}\mathbf{Y} = \mathbf{X} + \mathbf{N},\end{equation}where the noise vector $\mathbf{N} \in \mathbb{R}^{p}$ follows the $\mathrm{NDFHL}^{(d)}(\lambda, u)$ distribution. Consistent with the previous sections, we denote the effective dimension of the receiving hyperplane as $p = d-1$. Throughout, the input is subject to an average second-moment constraint $\mathbb{E}\|\mathbf{X}\|^2 \le P$. Our analysis focuses on the asymptotic regime $P \to \infty$, where we derive upper and lower bounds with matching pre-log factors to characterize the channel's degrees of freedom.

\subsection{Channel Capacity Formulation}

The channel capacity under a quadratic cost constraint is defined as:
\begin{equation}
\label{eq:capacity_def}
C(P) \triangleq \sup_{\substack{P_{\mathbf{X}}:~\mathbb{E}\|\mathbf{X}\|^2 \le P}} I(\mathbf{X}; \mathbf{X}+\mathbf{N}),
\end{equation}
where $\mathbf{X} \in \mathbb{R}^{p}$ is the input vector independent of the noise $\mathbf{N}$.
We assume the normalized drift speed satisfies $u > 0$. Under this condition, the mixing variable $T \sim IG(\lambda/u\sigma^2, \lambda^2/\sigma^2)$ admits finite logarithmic moments, ensuring that the differential entropy $h(\mathbf{N})$ is finite (see Appendix~\ref{app:integrability} for details).

Note that the singular case $u=0$, where the NDFHL distribution degenerates to a multivariate Cauchy distribution with infinite variance, falls outside the second-moment constraint framework. While our primary capacity scaling theorems focus on the finite-power regime ($u>0$), we explore the extension of our results towards this singular limit in Section \ref{sec:singular_limit}.

\subsection{Upper Bound}

We begin by establishing a universal upper bound based on the maximum entropy principle \cite{cover1999elements}. For any input $\mathbf{X}$ independent of $\mathbf{N}$, the mutual information is:
\begin{equation}
I(\mathbf{X}; \mathbf{Y}) = h(\mathbf{Y}) - h(\mathbf{N}).
\end{equation}
Recall from Proposition~\ref{prop:mean_zero} in Section \ref{sec:NDFHL_distribution} that the $\mathrm{NDFHL}^{(d)}$ noise is zero-mean, i.e., $\mathbb{E}[\mathbf{N}] = \mathbf{0}$. Consequently, the output second moment satisfies:
\begin{equation}
\mathbb{E}\|\mathbf{Y}\|^2 = \mathbb{E}\|\mathbf{X}\|^2 + \mathbb{E}\|\mathbf{N}\|^2 \le P + p \, \sigma_N^2,
\end{equation}
where $\sigma_N^2 = \lambda/u$ is the component-wise variance of the noise derived in Lemma~\ref{lem:variance}.

Among all random vectors with a fixed covariance trace, the \textit{spherical Gaussian distribution} maximizes differential entropy \cite{cover1999elements}. Thus, we can bound $h(\mathbf{Y})$ as:
\begin{align}
C(P) &= \sup_{P_{\mathbf{X}}} \{ h(\mathbf{Y}) - h(\mathbf{N}) \} \nonumber \\
&\le \frac{p}{2} \log\left( 2\pi e \frac{\mathbb{E}\|\mathbf{Y}\|^2}{p} \right) - h(\mathbf{N}) \nonumber \\
&\le \frac{p}{2} \log\left( 2\pi e \left( \frac{P}{p} + \sigma_N^2 \right) \right) - h(\mathbf{N}) \nonumber \\
&= \frac{p}{2} \log P + \frac{p}{2} \log \left( \frac{2\pi e}{p} + \frac{2\pi e \sigma_N^2}{P} \right) - h(\mathbf{N}). \label{eq:upper_bound_steps}
\end{align}
As $P \to \infty$, the second term in \eqref{eq:upper_bound_steps} converges to a constant. Therefore, we obtain the leading-order scaling:
\begin{equation}
\label{eq:upper_scaling}
C(P) \le \frac{p}{2}\log P + O(1), \qquad P \to \infty.
\end{equation}
This establishes that the pre-log factor (i.e., degrees of freedom) cannot exceed $p/2$, determined solely by the dimension of the receiving boundary.

\subsection{Achievable Lower Bound via GVM Conditioning}

To obtain a leading-order matching lower bound, we exploit the GVM structure of the noise established in Section \ref{sec:stochastic_foundation}. Specifically, $\mathbf{N} \stackrel{d}{=} \sigma \sqrt{T} \mathbf{Z}$, where $\mathbf{Z} \sim \mathcal{N}(\mathbf{0}, \mathbf{I}_p)$ and $T$ is an independent IG mixing variable.

The following lemma establishes a uniform bound on the \textit{conditioning gap} induced by a mixing variable $V>0$. This result serves as the key technical device to decouple the asymptotic capacity scaling from the latent temporal uncertainty, effectively isolating the geometric degrees of freedom from the semi-heavy-tailed noise structure.

\begin{lemma}[Uniform Bound on the Conditioning Gap]
\label{lem:conditioning_gap_mixing}
Consider the additive noise channel $\mathbf{Y} = \mathbf{X} + \mathbf{N}$ where $\mathbf{N} \stackrel{d}{=} \sqrt{V}\, \mathbf{Z}
$ is a scale mixture of Gaussian vectors ($\mathbf{Z} \sim \mathcal{N}(\mathbf{0}, \mathbf{I}_p)$ and $V > 0$ is independent of $\mathbf{Z}$ and $\mathbf{X}$). Then, for any input distribution $P_{\mathbf{X}}$, the conditioning gap is bounded by:
\begin{equation}
\label{eq:conditioning_gap_uniform_bound}
0 \le I(\mathbf{X}; \mathbf{Y}|V) - I(\mathbf{X}; \mathbf{Y}) \le I(V; \sqrt{V}\mathbf{Z}).
\end{equation}
Crucially, the upper bound $I(V; \sqrt{V}\mathbf{Z})$ depends solely on the distribution of the mixing variable $V$ and is independent of the input power $P$.
\end{lemma}

\begin{IEEEproof}
By the chain rule for mutual information, we can expand $I(\mathbf{X}; \mathbf{Y}, V)$ in two ways:
\begin{equation}
I(\mathbf{X}; \mathbf{Y}, V) = I(\mathbf{X}; \mathbf{Y}) + I(\mathbf{X}; V|\mathbf{Y}) = I(\mathbf{X}; \mathbf{Y}|V) + I(\mathbf{X}; V).
\end{equation}
Since the mixing variable $V$ is independent of the input $\mathbf{X}$, we have $I(\mathbf{X}; V) = 0$. This yields the identity:
\begin{equation}
I(\mathbf{X}; \mathbf{Y}|V) - I(\mathbf{X}; \mathbf{Y}) = I(\mathbf{X}; V|\mathbf{Y}).
\end{equation}
To bound this gap, we observe:
\begin{align}
I(\mathbf{X}; V|\mathbf{Y}) &\le I(\mathbf{X}, \mathbf{Y}; V) \nonumber \\
&= I(\mathbf{X}; V) + I(\mathbf{Y}; V|\mathbf{X}) \nonumber \\
&= I(\mathbf{Y}; V|\mathbf{X}).
\end{align}
Given $\mathbf{X}=\mathbf{x}$, the output is $\mathbf{Y} = \mathbf{x} + \sqrt{V}\mathbf{Z}$. By translation invariance, $I(\mathbf{x}+\sqrt{V}\mathbf{Z}; V) = I(\sqrt{V}\mathbf{Z}; V)$. 
Averaging over the distribution of $\mathbf X$ yields
\begin{equation}
\label{eq:conditional_translation_invariance}
I(\mathbf Y;V | \mathbf X)
=
\mathbb E_{\mathbf X}\!\left[
I(\mathbf Y;V | \mathbf X=\mathbf x)
\right]
=
I(\sqrt{V}\mathbf Z;V).
\end{equation}
Thus we have
\begin{equation}
I(\mathbf{X}; V|\mathbf{Y}) \le I(\sqrt{V}\mathbf{Z}; V).
\end{equation}
This bound is intrinsic to the noise structure and invariant to $P$, concluding the proof.
\end{IEEEproof}

We now apply this lemma to the NDFHL channel, where the mixing variance is $V = \sigma^2 T$.


\begin{lemma}[Achievable Scaling with Gaussian Input]
\label{lem:GVM_lower}
Let $\mathbf{N} \sim \mathrm{NDFHL}^{(d)}(\lambda, u)$ with $u > 0$. For a Gaussian input distribution
\begin{equation}
\label{eq:Gaussian_input_def}
\mathbf{X} \sim \mathcal{N}\left(\mathbf{0}, \frac{P}{p} \mathbf{I}_{p}\right),
\end{equation}
the mutual information satisfies:
\begin{equation}
\label{eq:lower_scaling}
I(\mathbf{X}; \mathbf{Y}) \ge \frac{p}{2}\log P + O(1), \qquad P \to \infty.
\end{equation}
\end{lemma}

\begin{IEEEproof}

First we note that the isotropic Gaussian input $\mathbf{X}$ defined in \eqref{eq:Gaussian_input_def} satisfies the second-moment constraint with equality. Specifically, the total average power is given by $\mathbb{E}\|\mathbf{X}\|^2 = \operatorname{tr}(\frac{P}{p} \mathbf{I}_p) = p \cdot \frac{P}{p} = P$ . This ensures that the proposed input distribution remains feasible for the capacity optimization problem \eqref{eq:capacity_def} for any $P > 0$.

Applying Lemma \ref{lem:conditioning_gap_mixing} with the specific NDFHL mixing variable $V = \sigma^2 T$, we obtain the following decomposition:
\begin{equation}
    I(\mathbf{X}; \mathbf{Y}) \ge I(\mathbf{X}; \mathbf{Y} \mid T) - I(T; \sigma\sqrt{T}\mathbf{Z}).
    \label{eq:lower_bound_expansion}
\end{equation}
The gap term $I(T; \sigma\sqrt{T}\mathbf{Z})$ is independent of the input power $P$ and depends solely on the channel parameters. By Lemma \ref{lem:finite_MI} in Appendix~\ref{app:integrability}, the condition $u > 0$ ensures that $\mathbb{E}[T]$ and $\mathbb{E}[|\log T|]$ are finite, which in turn guarantees that $I(T; \sigma\sqrt{T}\mathbf{Z}) < \infty$. Consequently, this second term remains a constant $O(1)$ penalty as $P \to \infty$. 
We now focus on the first term. 

\emph{Recall:} For an additive Gaussian channel $\mathbf Y=\mathbf X+\mathbf N$ with
$\mathbf X\sim\mathcal N(\mathbf 0,\Sigma_X)$ independent of
$\mathbf N\sim\mathcal N(\mathbf 0,\Sigma_N)$, the output is Gaussian with
$\mathbf Y\sim\mathcal N(\mathbf 0,\Sigma_X+\Sigma_N)$, and
\begin{equation}
I(\mathbf X;\mathbf Y)=h(\mathbf Y)-h(\mathbf N)
=\frac{1}{2}\log\frac{|\Sigma_X+\Sigma_N|}{|\Sigma_N|}.
\end{equation}
Conditioning on $T=t$, we have $\Sigma_X=\frac{P}{p}\mathbf{I}_p$ and
$\Sigma_N=\sigma^2 t\,\mathbf{I}_p$, hence
\begin{align}
\begin{split}
I(\mathbf X;\mathbf Y\mid T=t)
&=\frac{1}{2}\log\frac{\left|\left(\frac{P}{p}+\sigma^2 t\right)\mathbf{I}_p\right|}
{\left|\sigma^2 t\,\mathbf{I}_p\right|}\\
&=\frac{1}{2}\log\frac{\left(\frac{P}{p}+\sigma^2 t\right)^p}{(\sigma^2 t)^p}
=\frac{p}{2}\log\!\left(1+\frac{P/p}{\sigma^2 t}\right).
\end{split}
\end{align}

Averaging over $T$, we obtain
\begin{align}
I(\mathbf{X}; \mathbf{Y} \mid T) &= \mathbb{E}\left[ \frac{p}{2} \log\left(1 + \frac{P}{p \sigma^2 T}\right) \right] \nonumber \\
&= \frac{p}{2} \log P - \frac{p}{2} \log(p\sigma^2) - \frac{p}{2} \mathbb{E}[\log T] \nonumber \\
&\quad + \frac{p}{2} \mathbb{E}\left[ \log\left(1 + \frac{p \sigma^2 T}{P}\right) \right]. \label{eq:expansion}
\end{align}
For $u > 0$, $T$ follows a proper Inverse Gaussian distribution, implying that $\mathbb{E}[\log T]$ is finite (see Appendix~\ref{app:integrability}). 
Regarding the last term in \eqref{eq:expansion}, we rigorously justify its convergence to zero as $P \to \infty$ via the Dominated Convergence Theorem (DCT). 

First, for any fixed realization of the hitting time $T=t$, the term $p \sigma^2 t / P$ vanishes as $P$ increases, implying the pointwise convergence of the integrand:
\begin{equation}
\lim_{P \to \infty} \frac{p}{2} \log\left(1 + \frac{p \sigma^2 t}{P}\right) = \frac{p}{2} \log(1) = 0.
\end{equation}
Second, to apply the DCT, we establish an integrable dominating function. Using the inequality $\log(1+x) \le x$ for $x \ge 0$, the integrand is bounded for any $P \ge P_{\min} > 0$ as:
\begin{equation}
\left| \frac{p}{2} \log\left(1 + \frac{p \sigma^2 T}{P}\right) \right| \le \frac{p^2 \sigma^2}{2 P} T \le \frac{p^2 \sigma^2}{2 P_{\min}} T \triangleq h(T).
\end{equation}
As shown in Appendix \ref{app:integrability}, the first moment $\mathbb{E}[T] = \lambda/(u\sigma^2)$ is finite for $u > 0$. Thus, $h(T)$ is an integrable function since $\mathbb{E}[h(T)] = \frac{p^2 \lambda}{2 u P_{\min}} < \infty$. By the DCT, we can exchange the limit and the expectation:
\begin{align}
\begin{split}
\lim_{P \to \infty} &\mathbb{E}\left[ \frac{p}{2} \log\left(1 + \frac{p \sigma^2 T}{P}\right) \right] \\
=~&\mathbb{E}\left[ \lim_{P \to \infty} \frac{p}{2} \log\left(1 + \frac{p \sigma^2 T}{P}\right) \right] = 0.
\end{split}
\end{align}
Consequently, the conditional mutual information scales as $I(\mathbf{X}; \mathbf{Y} | T) = \frac{p}{2} \log P + O(1)$, where the $O(1)$ terms are independent of $P$. Combining this with \eqref{eq:lower_bound_expansion} completes the proof.

\end{IEEEproof}

\subsection{Main Capacity Scaling Result}

Combining the upper bound \eqref{eq:upper_scaling} and the achievable lower bound from Lemma \ref{lem:GVM_lower}, we can now state our main scaling result.

\begin{theorem}[Capacity Scaling of the NDFHL Channel]
\label{thm:main_scaling}
Consider the additive noise channel $\mathbf{Y} = \mathbf{X} + \mathbf{N}$ where the noise $\mathbf{N} \sim \mathrm{NDFHL}^{(d)}(\lambda, u)$ has normalized normal drift $u > 0$ and dimension $p = d-1$. Under the average power constraint $\mathbb{E}\|\mathbf{X}\|^2 \le P$, the channel capacity scales as:
\begin{equation}
\label{eq:final_scaling_result}
C(P) = \frac{p}{2} \log P + O(1), \quad \text{as } P \to \infty.
\end{equation}
\end{theorem}

\begin{IEEEproof}
The universal (i.e., among all input distributions) upper bound is given by \eqref{eq:upper_scaling}. The lower bound is achieved by the isotropic Gaussian input defined in Lemma \ref{lem:GVM_lower}. Since the pre-log factors match, the theorem is proved.
\end{IEEEproof}


\subsection{Refining Asymptotic Bounds}
\label{subsec:refined_scaling}While Theorem \ref{thm:main_scaling} establishes the pre-log scaling factor, the methodology developed in the previous subsections allows for a more granular characterization of the capacity. Specifically, by carefully collecting the non-vanishing terms from the upper and lower bound derivations, we obtain a refined expansion of the high-SNR offset.

\begin{theorem}[Asymptotic Capacity Exactness]\label{thm:refined_scaling}For the NDFHL additive noise channel with $u > 0$, the capacity $C(P)$ admits the following refined expansion as $P \to \infty$:\begin{equation}\label{eq:refined_scaling_result}C(P) = \frac{p}{2} \log P + c^* + o(1),\end{equation}where the constant offset $c^*$ is uniquely given by:\begin{equation}\label{eq:exact_offset}c^* = \frac{p}{2} \log\left( \frac{2\pi e}{p} \right) - h(\mathbf{N}).\end{equation}
\end{theorem}

\begin{IEEEproof}
We establish the exactness of \eqref{eq:exact_offset} by showing that the universal upper bound and the Gaussian-input lower bound converge to the same constant.

\textit{1) Upper Bound:} From the maximum entropy principle discussed in \eqref{eq:upper_bound_steps}, the capacity is bounded by $C(P) \le \frac{p}{2} \log(2\pi e (\frac{P}{p} + \sigma_N^2)) - h(\mathbf{N})$. Expanding for large $P$:
\begin{equation}
C(P) \le \frac{p}{2} \log P + \frac{p}{2} \log ( 2\pi e / p ) - h(\mathbf{N}) + O(P^{-1}),
\end{equation}
identifying $c_{\mathrm{UB}} = \frac{p}{2} \log ( 2\pi e/p ) - h(\mathbf{N})$.

\textit{2) Lower Bound:} Applying the GVM decomposition from Lemma \ref{lem:GVM_lower} and the expansion in \eqref{eq:expansion}, the achievable rate satisfies:
\begin{equation}
\label{eq:lower_bound_slim}
\begin{aligned}
I(\mathbf{X}; \mathbf{Y}) \ge \frac{p}{2} \log P &- \frac{p}{2} \log(p \sigma^2) \\ 
&- \frac{p}{2}\mathbb{E}[\log T] - I(T; \mathbf{N}) + o(1),
\end{aligned}
\end{equation}
identifying $c_{\mathrm{LB}} = \frac{p}{2} \log ( 1/(p\sigma^2) ) - \frac{p}{2} \mathbb{E}[\log T] - I(T; \mathbf{N})$.

\textit{3) Convergence:} To see $c_{\mathrm{LB}} = c_{\mathrm{UB}}$, recall $h(\mathbf{N}) = h(\mathbf{N} \mid T) + I(\mathbf{N}; T)$. Substituting $h(\mathbf{N} \mid T) = \frac{p}{2} \log(2\pi e \sigma^2) + \frac{p}{2} \mathbb{E}[\log T]$ into $c_{\mathrm{UB}}$:
\begin{equation}
\begin{aligned}
c_{\mathrm{UB}} &= \frac{p}{2} \log( 2\pi e / p ) - \bigl[ \tfrac{p}{2} \log(2\pi e \sigma^2) + \tfrac{p}{2} \mathbb{E}[\log T] + I(\mathbf{N}; T) \bigr] \\
&= \frac{p}{2} \bigl[ \log( 2\pi e / p ) - \log(2\pi e \sigma^2) \bigr] - \tfrac{p}{2} \mathbb{E}[\log T] - I(\mathbf{N}; T) \\
&= \frac{p}{2} \log( 1 / p \sigma^2 ) - \frac{p}{2} \mathbb{E}[\log T] - I(\mathbf{N}; T) = c_{\mathrm{LB}}.
\end{aligned}
\end{equation}
Since the upper and lower bounds coincide at the constant term, the expansion is exact.
\end{IEEEproof}

\subsection{Implications of the Exact Scaling}
\label{subsec:implications}

The established exactness of the capacity expansion ($c_{\mathrm{LB}} = c_{\mathrm{UB}}$) leads to a profound information-theoretic conclusion: the asymptotic ``shaping gap'' for the NDFHL channel is exactly zero. This yields several key implications for boundary-induced hitting position-based communication systems:

\subsubsection{Asymptotic Optimality of Gaussian Signaling}
The vanishing gap implies that the isotropic Gaussian distribution $\mathbf{X} \sim \mathcal{N}(\mathbf{0}, \frac{P}{p}\mathbf{I}_p)$ is asymptotically capacity-achieving. Despite the non-Gaussian, semi-heavy-tailed nature of the physical noise $\mathbf{N}$, standard Gaussian codebooks incur no first-order rate loss as $P \to \infty$. This stands in sharp contrast to other non-Gaussian channels, such as the uniform noise channel, which incurs a shaping penalty of $\frac{1}{2} \log(2\pi e/12) \approx 0.255$ bits per dimension \cite{cover1999elements,Forney:1998}. The result suggests that standard coding schemes designed for AWGN channels can be directly applied to these molecular transport systems without a fundamental loss in degrees of freedom.

\subsubsection{Entropy-Dominant Scaling}
The capacity offset is determined exclusively by the differential entropy of the noise, $h(\mathbf{N})$. The expansion $C(P) \approx \frac{p}{2} \log(P) + [\frac{p}{2} \log(2\pi e/p) - h(\mathbf{N})]$ indicates that the channel behaves asymptotically like an equivalent AWGN channel with effective noise power:
\begin{equation}
P_{N}^{\mathrm{eff}} \triangleq \frac{1}{2\pi e} \exp\left( \frac{2}{p} h(\mathbf{N}) \right).
\end{equation}
This ``Entropy Power'' characterization confirms that while the drift-diffusion dynamics ($\lambda, u, \sigma$) generate a complex spatial distribution, their impact on transmission limits is fully encapsulated by the aggregate spatial uncertainty $h(\mathbf{N})$ they induce.

\subsubsection{Geometric Universality and Invariance}
At high SNR, the aggregate spatial uncertainty dominates the latent temporal fluctuations of $T$. Consequently, the zero-gap property holds uniformly for all $u > 0$ and $d \ge 2$, identifying Gaussian optimality as an intrinsic geometric feature of the first-hitting mechanism. This invariance ensures that the capacity scaling remains robust as the transport dynamics transition from the ballistic to the diffusion-dominated regime.


\section{The Singular Limit: Connection to Additive Cauchy Noise}
\label{sec:singular_limit}

Throughout this section, the discussion is interpretative.
All formal capacity theorems in this paper are established under the assumption
$u>0$, and no exchange of limits between $P\to\infty$ and $u\to 0$ is claimed.
Nevertheless, the exact entropy-dominant scaling law established for all fixed
$u>0$ provides a natural reference point for examining the singular limit in which
the normalized drift vanishes.

The capacity scaling law derived in Theorems~\ref{thm:main_scaling}
and~\ref{thm:refined_scaling} relies on the assumption $u>0$, which ensures that
the mixing variable $T$ admits finite moments and that the boundary-induced noise
$\mathbf N$ has finite covariance $\E\|\mathbf N\|^2=p\lambda/u$.
However, the underlying geometric construction suggests a continuous structural
transition as the drift strength decreases.
In particular, as $u\downarrow 0$, the $\mathrm{NDFHL}^{(d)}$ family degenerates to
a heavy-tailed distribution with infinite variance, while retaining a
well-defined and finite differential entropy, see \eqref{eq:gp_def}.
This observation motivates a comparison with the additive white Cauchy noise
(AWCN) model recently analyzed by Pang and Zhang~\cite{pang2025information},
and suggests that both regimes may be unified under a common entropy-based
geometric framework.

\subsection{Review of Capacity Scaling for Cauchy Noise}

The AWCN channel highlights a fundamental limitation of classical
signal-to-noise-ratio (SNR) based intuition.
For additive Cauchy noise, the variance is infinite, rendering the conventional
SNR identically zero for any finite input power constraint.
Despite this energetic singularity, Pang and Zhang~\cite{pang2025information}
established that the capacity of the scalar AWCN channel ($p=1$), subject to a
second-moment constraint $P$, still exhibits logarithmic growth with $P$.

Specifically, for noise $Z\sim\mathcal C(0,\lambda)$ with scale parameter $\lambda$,
the high-power capacity satisfies
\begin{equation}
\label{eq:pang_result_v3}
C_{\mathrm{Cauchy}}(P)
=
\frac{1}{2}\log\!\left(\frac{P}{\lambda^2}\right)
+
\text{const.}
+
o(1),
\qquad P\to\infty,
\end{equation}
where the constant term is bounded between two finite values.
The lower bound is achieved using Gaussian signaling together with the entropy
power inequality (EPI) \cite{shannon1948mathematical,Dembo:1991}, while the upper bound is obtained via a genie-aided
argument \cite{pang2025information} that renders the channel conditionally Gaussian.
This result demonstrates that the divergence of the noise variance does not, by
itself, impose a fundamental information-theoretic obstruction, and that the
degrees of freedom of the channel remain preserved.

\subsection{The Vanishing Drift Limit \texorpdfstring{$u\downarrow 0$}{u to 0}}

We now relate the $\mathrm{NDFHL}^{(d)}$ model to the Cauchy setting.
Recall from Section~\ref{sec:stochastic_foundation} that the mixing variable
$T$ follows the IG distribution
$IG(\lambda/\mu\sigma^2,\lambda^2/\sigma^2)$.

\begin{proposition}[Distributional Convergence]
\label{prop:cauchy_convergence}
As the normalized drift vanishes ($u\downarrow 0$), the mixing variable $T$
converges in distribution to the L\'evy distribution (a stable law \cite{Feller:1991} with index
$1/2$).
Consequently, the boundary-induced noise
$\mathbf N\sim\mathrm{NDFHL}^{(d)}(\lambda,u)$ converges in distribution to the
isotropic multivariate Cauchy distribution with scale parameter $\lambda$.
\end{proposition}

\begin{IEEEproof}[Proof sketch]
Recall that the first-passage time $T$ admits the IG law
$T\sim IG(\nu,\kappa)$ with mean $\nu=\lambda/\mu$ and shape
$\kappa=\lambda^2/\sigma^2$ as in \eqref{eq:IG_params}.
A convenient characterization is via the Laplace transform of the
$IG(\nu,\kappa)$ distribution:
\begin{equation}
\label{eq:IG_Laplace_nu_kappa}
\E\!\left[e^{-sT}\right]
=
\exp\!\left(
\frac{\kappa}{\nu}
\Bigl(
1-\sqrt{1+\frac{2\nu^2 s}{\kappa}}
\Bigr)
\right),
\qquad s\ge 0,
\end{equation}
see, e.g., \cite{Feller:1991}.
Consider the vanishing-drift limit $u\downarrow 0$, which corresponds to
$\mu\downarrow 0$ and hence $\nu=\lambda/\mu\to\infty$, while $\kappa$ remains
fixed.
Then, for any fixed $s\ge 0$,
\[
\sqrt{1+\frac{2\nu^2 s}{\kappa}}
=
\frac{\nu\sqrt{2s}}{\sqrt{\kappa}}\,(1+o(1)),
\qquad \nu\to\infty,
\]
and therefore
\[
\frac{\kappa}{\nu}
\Bigl(
1-\sqrt{1+\frac{2\nu^2 s}{\kappa}}
\Bigr)
=
-\sqrt{2\kappa s}+o(1).
\]
Substituting this into \eqref{eq:IG_Laplace_nu_kappa} yields
\begin{equation}
\label{eq:Levy_Laplace_kappa}
\E\!\left[e^{-sT}\right]
\longrightarrow
\exp\!\left(-\sqrt{2\kappa s}\right),
\qquad s\ge 0,
\end{equation}
which is the Laplace transform of the (one-sided) L\'evy distribution, i.e.,
the $1/2$-stable law. This proves the claimed distributional convergence of $T$.

Next, recall the GVM representation of the boundary noise:
\[
\mathbf N \stackrel{d}{=} \sigma\sqrt{T}\,\mathbf Z,
\qquad
\mathbf Z\sim\mathcal N(\mathbf 0,\mathbf{I}_p),
\]
with $\mathbf Z$ independent of $T$.
Conditioned on $T$, we have $\mathbf N\mid T\sim\mathcal N(\mathbf 0,\sigma^2 T I_p)$,
hence the characteristic function is
\[
\varphi_{\mathbf N}(\boldsymbol\omega)
=
\E\!\left[
\exp\!\left(
-\frac{\sigma^2 T}{2}\|\boldsymbol\omega\|^2
\right)
\right]
=
\E\!\left[e^{-sT}\right]_{\,s=\frac{\sigma^2}{2}\|\boldsymbol\omega\|^2}.
\]
Applying \eqref{eq:Levy_Laplace_kappa} with
$s=\frac{\sigma^2}{2}\|\boldsymbol\omega\|^2$ gives
\[
\varphi_{\mathbf N}(\boldsymbol\omega)
\longrightarrow
\exp\!\left(
-\sqrt{2\kappa\cdot\frac{\sigma^2}{2}\|\boldsymbol\omega\|^2}
\right)
=
\exp\!\left(-\lambda\|\boldsymbol\omega\|\right),
\]
where we used $\kappa=\lambda^2/\sigma^2$.
The limit $\exp(-\lambda\|\boldsymbol\omega\|)$ is the characteristic function of an
isotropic multivariate Cauchy distribution with scale parameter $\lambda$.
The claim follows.
\end{IEEEproof}

\begin{figure}[!t]
    \centering
    \includegraphics[width=1.0\linewidth]{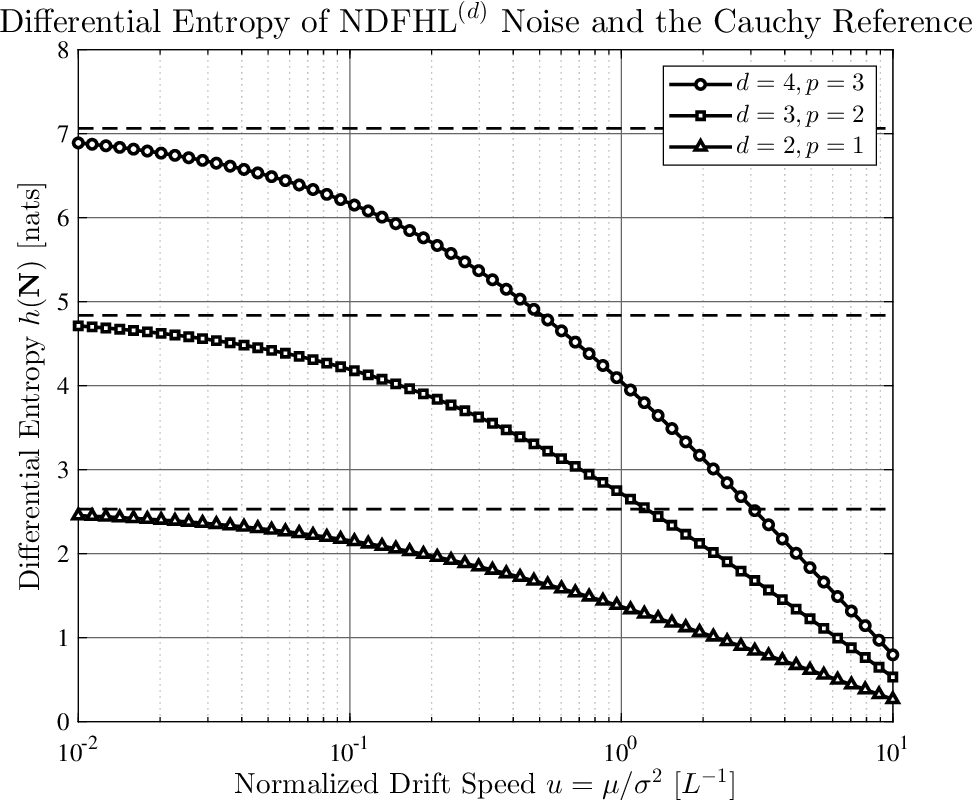}
    \caption{Differential entropy $h(\mathbf{N})$ of the $\mathrm{NDFHL}^{(d)}$
    distribution as a function of the normalized drift speed $u$ for dimensions
    $d=2,3,4$ with $\lambda=1$.
    Solid markers correspond to numerical integration results, while the dashed
    horizontal lines indicate the differential entropy of the corresponding
    isotropic multivariate Cauchy distributions obtained in the limit $u=0$.
    For all $u>0$, the entropy $h(\mathbf{N})$ is finite and is observed to approach,
    as $u \downarrow 0$, the Cauchy entropy value $g(p)$, see \eqref{eq:gp_def}, even
    though the second moment of the noise diverges in this limit.
    }
    \label{fig:entropy_convergence}
\end{figure}

In the limit procedure stated in Proposition~\ref{prop:cauchy_convergence}, the first moment $\E[T]=\lambda/u$ diverges, and the noise covariance
$\mathbf K_{\mathbf N}=(\lambda/u)\mathbf{I}_{d-1}$ becomes unbounded.
Nevertheless, numerical evaluation (see Fig.~\ref{fig:entropy_convergence})
indicates that the differential entropy $h(\mathbf N)$ remains finite and varies
continuously with respect to $u$.
Combining this observation with the refined high-SNR expansion in
Theorem~\ref{thm:refined_scaling} for fixed $u>0$ suggests that the leading-order
capacity scaling persists as
\begin{equation}
C_{\mathrm{NDFHL}}(P)
\approx
\frac{p}{2}\log P + \mathrm{const.},
\end{equation}
even as the system approaches the singular, infinite-variance regime.

For reference, the differential entropy of the isotropic multivariate Cauchy
distribution in $p$ dimensions admits a closed-form expression.
Let $\mathbf C_p$ denote a $p$-dimensional isotropic Cauchy random vector with scale
parameter $\lambda>0$, whose density is \cite{Verdu_Entropy23}
\begin{equation}
f_{\mathbf C_p}(\mathbf x)
=
\frac{\Gamma\!\left(\frac{p+1}{2}\right)}
{\pi^{\frac{p+1}{2}}\lambda^{p}}
\left(1+\frac{\|\mathbf x\|^2}{\lambda^2}\right)^{-\frac{p+1}{2}}.
\end{equation}
Its differential entropy is given by
\begin{equation}
h(\mathbf C_p)
=
g(p),
\end{equation}
where
\begin{equation}
\label{eq:gp_def}
g(p)
=
\log\!\Bigg(
\frac{\pi^{\frac{p+1}{2}}\lambda^{p}}
{\Gamma\!\left(\frac{p+1}{2}\right)}
\Bigg)
+
\frac{p+1}{2}
\Bigg[
\psi\!\left(\frac{p+1}{2}\right)
-
\psi\!\left(\frac12\right)
\Bigg],
\end{equation}
and $\psi(\cdot)$ denotes the digamma function.
In particular, for $p=1$ (corresponding to $d=2$), this reduces to
$g(1)=\log(4\pi\lambda)$.

\subsection{Entropy-Dominant Offset and Structural Continuity}

To quantify the connection with the AWCN benchmark, we define the asymptotic
capacity offset
\begin{equation}
\mathcal L(u)
\triangleq
\lim_{P\to\infty}
\left(
C(P)-\frac{p}{2}\log P
\right).
\end{equation}
From the refined expansion in Section~\ref{subsec:refined_scaling}, this offset is
entirely determined by the noise entropy,
\begin{equation}
\label{eq:offset_NDFHL}
\mathcal L(u)
=
\frac{p}{2}\log\!\left(\frac{2\pi e}{p}\right)
-
h(\mathbf N).
\end{equation}
For $p=1$ and $u\downarrow 0$, the noise converges to a scalar Cauchy random
variable with entropy $h(N)=g(1)=\log(4\pi\lambda)$, yielding
\begin{equation}
\lim_{u\downarrow 0}\mathcal L(u)
=
\frac{1}{2}\log(2\pi e)-\log(4\pi\lambda).
\end{equation}
This expression precisely recovers the capacity lower bound derived by Pang and Zhang in~\cite[Eq.~(8)]{pang2025information}, which corresponds to the achievable rate of AWCN using Gaussian inputs.
This agreement confirms that our geometric scaling law (Theorem~\ref{thm:main_scaling} and Theorem~\ref{thm:refined_scaling}) smoothly bridges the finite-variance and infinite-variance regimes, capturing the correct logarithmic dependence on the scale parameter $\lambda$ even at the singularity.


\section{Conclusion}
\label{sec:conclusion}

This work presented a rigorous stochastic and information-theoretic analysis of
a class of boundary-induced additive noise channels generated by
drift--diffusion processes.
By identifying the $\mathrm{NDFHL}^{(d)}$ distribution as a Gaussian
variance--mixture, we shifted the analytical focus from intractable
Bessel-type density expressions to the underlying transport mechanism on path
space.
This perspective enabled a transparent characterization of the channel in
terms of latent Gaussian structure and geometric observables associated with
first-hitting events.

\emph{First}, at the technical level, we established an exact high-SNR capacity
expansion under a second-moment constraint.
Beyond identifying the pre-log factor, we showed that the upper and lower bounds
match at the constant level, yielding a vanishing asymptotic gap.
As a consequence, isotropic Gaussian signaling is asymptotically
capacity-achieving for all fixed drift strengths $u>0$ and dimensions
$d\ge2$.
This zero-gap result demonstrates that, despite the non-Gaussian and
semi-heavy-tailed nature of the boundary-induced noise, no first-order rate
penalty is incurred by standard Gaussian codebooks in the high-power regime.

\emph{Second}, the refined expansion reveals an entropy-dominant universality
principle.
While the logarithmic growth is dictated by the input power, the entire
constant offset is governed solely by the differential entropy of the noise.
All geometric and physical parameters of the transport process---including the
boundary separation, drift strength, and diffusion coefficient---enter the
capacity only through the aggregate spatial uncertainty they induce.
In this sense, the high-SNR behavior of boundary-hitting channels is controlled
by geometry and entropy rather than energetic moments, placing the
$\mathrm{NDFHL}$ family within a broader class of channels that admit an
entropy-power interpretation.

\emph{Finally}, we explored the singular vanishing-drift limit as an
interpretative extension of the theory.
As the normalized drift approaches zero, the $\mathrm{NDFHL}$ distribution
degenerates to a heavy-tailed Cauchy law with infinite variance.
Although this regime lies outside the formal scope of the capacity theorems,
numerical evidence and structural continuity suggest that the entropy-dominant
scaling persists, connecting the finite-variance boundary-noise model with the
additive white Cauchy noise channel studied in recent work.
This observation points to a broader geometric robustness of first-hitting
location observables and highlights the potential of entropy-based methods for
analyzing singular stochastic transport limits.

Taken together, these results position geometry-induced Gaussian
variance--mixture channels as a natural and analytically tractable benchmark
for non-Gaussian additive noise models arising from stochastic transport.
They also suggest several promising directions for future work, including a
rigorous treatment of singular limits, extensions to more general boundary
geometries, and the interplay between entropy-dominant scaling and statistical
inference on path space.

\appendices


\section{Integrability and Entropy Finiteness of the NDFHL Distribution}
\label{app:integrability}

This appendix provides the rigorous mathematical justification for the finiteness of the key information-theoretic quantities appearing in the capacity scaling analysis. 
We first establish the existence of the logarithmic moments for the IG mixing variable $T$. 
Building on these results, we proceed to prove the finiteness of the differential entropy $h(\mathbf{N})$ and the mutual information $I(T; \mathbf{N})$ for the NDFHL family. 
These properties ensure that the constant terms $c_{\mathrm{LB}}$ and $c_{\mathrm{UB}}$ derived in Theorem \ref{thm:refined_scaling} are well-defined and finite.

\subsection{Inverse Gaussian Preliminaries}

Recall from Section \ref{sec:stochastic_foundation} that the first-hitting time $T$ follows the IG distribution, $T \sim IG(\nu, \kappa)$, with parameters determined by the physical channel geometry:
\begin{equation}
\label{eq:IG_phys_params}
\nu = \frac{\lambda}{\mu}, \qquad \kappa = \frac{\lambda^2}{\sigma^2}.
\end{equation}
Assuming positive drift ($\mu > 0$), both $\nu$ and $\kappa$ are strictly positive. The PDF is given by:
\begin{equation}
\label{eq:IG_density_T}
f_T(t) = \sqrt{\frac{\kappa}{2\pi t^3}} \exp\left( -\frac{\kappa (t - \nu)^2}{2\nu^2 t} \right), \quad t > 0.
\end{equation}
To facilitate the analysis of integrability, we rewrite the exponent as:
\begin{equation}
\frac{\kappa (t - \nu)^2}{2\nu^2 t} = \frac{\kappa t^2 - 2\kappa t \nu + \kappa \nu^2}{2\nu^2 t} = \frac{\kappa t}{2\nu^2} - \frac{\kappa}{\nu} + \frac{\kappa}{2t}.
\end{equation}
Thus, the density takes the form:
\begin{equation}
\label{eq:IG_density_expanded}
f_T(t) = C_{\nu,\kappa} \cdot t^{-3/2} \cdot \exp\left( -\frac{\kappa}{2\nu^2} t - \frac{\kappa}{2t} \right),
\end{equation}
where $C_{\nu,\kappa} = \sqrt{\frac{\kappa}{2\pi}} e^{\kappa/\nu}$.

\subsection{Logarithmic Moment Finiteness of IG}

The finiteness of the logarithmic moment $\mathbb{E}[\log T]$ is a necessary condition for the high-SNR expansion established in Lemma~\ref{lem:GVM_lower}. In this subsection, we provide a generalized integrability result for the IG variable $T$, demonstrating that both its logarithmic and inverse moments of arbitrary order are finite for all $u > 0$.

\begin{lemma}[Finiteness of Moments]
\label{lem:log_moment}
Let $T \sim IG(\nu, \kappa)$ with $\nu, \kappa > 0$. Then:
\begin{enumerate}
    \item The first inverse moment is finite: $\mathbb{E}[T^{-1}] = \nu^{-1} + \kappa^{-1} < \infty$.
    \item The mean absolute logarithm is finite: $\mathbb{E}\bigl[ |\log T| \bigr] < \infty$.
\end{enumerate}
\end{lemma}

\begin{IEEEproof}
The finiteness of $\mathbb{E}[T^{-1}]$ is a standard property of the Inverse Gaussian distribution. We focus on establishing $\mathbb{E}[|\log T|] < \infty$. We partition the expectation integral into two regimes: the origin $(0, 1]$ and the tail $(1, \infty)$.

\emph{(i) Behavior near the origin ($0 < t \le 1$):}
Using the bound $|\log t| = -\log t$ for $t \in (0, 1]$, the contribution to the expectation is:
\begin{equation}
I_0 = \int_0^1 (-\log t) f_T(t) \, \diff t.
\end{equation}
From \eqref{eq:IG_density_expanded}, noting that $e^{-\frac{\kappa}{2\nu^2}t} \le 1$ on this interval, we have:
\begin{equation}
f_T(t) \le C_{\nu,\kappa} \, t^{-3/2} e^{-\frac{\kappa}{2t}}.
\end{equation}
Substituting $x = 1/t$, the integral transforms to:
\begin{equation}
I_0 \le C_{\nu,\kappa} \int_1^\infty (\log x) \, x^{-1/2} e^{-\frac{\kappa}{2} x} \, \diff x.
\end{equation}
Since $\log x = o(x^\epsilon)$ as $x \to \infty$ for any $\epsilon>0$, there exists
a constant $C>0$ such that
\begin{equation}
(\log x)\,x^{-1/2} \;\le\; C x,
\qquad \text{for all } x \ge 1.
\end{equation}
Consequently,
\begin{equation}
(\log x)\,x^{-1/2} e^{-(\kappa/2)x}
\;\le\;
C x e^{-(\kappa/2)x},
\end{equation}
and the right-hand side is integrable over $[1,\infty)$.
Therefore, the integral $I_0$ is finite.

\emph{(ii) Behavior at the tail ($t > 1$):}
Here $|\log t| = \log t < t$ for $t$ sufficiently large. The contribution is:
\begin{equation}
I_\infty = \int_1^\infty (\log t) f_T(t) \, \diff t \le \int_1^\infty t f_T(t) \, \diff t \le \mathbb{E}[T].
\end{equation}
Since $\mathbb{E}[T] = \nu < \infty$, the tail integral $I_\infty$ converges.

Combining both regimes, $\mathbb{E}[|\log T|] = I_0 + I_\infty < \infty$.
\end{IEEEproof}

\subsection{Finiteness of Differential Entropy of IG}

\begin{lemma}[Finite Differential Entropy of $T$]
\label{lem:entropy_T}
The differential entropy $h(T)$ is finite.
\end{lemma}

\begin{IEEEproof}
The differential entropy is defined as $h(T) = -\mathbb{E}[\log f_T(T)]$. Taking the logarithm of the density expression in \eqref{eq:IG_density_expanded}:
\begin{equation}
\log f_T(T) = \log C_{\nu,\kappa} - \frac{3}{2} \log T - \frac{\kappa}{2\nu^2} T - \frac{\kappa}{2} T^{-1}.
\end{equation}
Taking the expectation:
\begin{equation}
h(T) = -\log C_{\nu,\kappa} + \frac{3}{2} \mathbb{E}[\log T] + \frac{\kappa}{2\nu^2} \mathbb{E}[T] + \frac{\kappa}{2} \mathbb{E}[T^{-1}].
\end{equation}
From Lemma \ref{lem:log_moment}, $\mathbb{E}[\log T]$ and $\mathbb{E}[T^{-1}]$ are finite. Since $\mathbb{E}[T] = \nu$ is also finite, every term on the RHS is bounded. Thus, $h(T)$ is finite (i.e., $h(T)\in\R$).
\end{IEEEproof}

\subsection{Finiteness of the NDFHL Entropy}
\label{app:entropy_finite}

We confirm that the differential entropy of the boundary-induced noise $\mathbf{N}$ is finite, ensuring that the constant term $c_{\mathrm{UB}}$ derived in Theorem~\ref{thm:refined_scaling} is meaningful.

\begin{lemma}
\label{prop:noise_entropy_finite}
Let $\mathbf{N} \sim \mathrm{NDFHL}^{(d)}(\lambda, u)$ with $\lambda, u > 0$. Then $|h(\mathbf{N})| < \infty$.
\end{lemma}

\begin{IEEEproof}
Recall from Lemma~\ref{lem:entropy_bounds} in Section~\ref{sec:properties} that the differential entropy is bounded by:
\begin{equation}
\underline{h}(u) \le h(\mathbf{N}) \le \bar{h}(u).
\end{equation}
We examine the finiteness of these bounds:
\begin{itemize}
    \item \textbf{Upper Bound:} The upper bound is given by $\bar{h}(u) = \frac{p}{2} \log(2\pi e \frac{\lambda}{u})$. Since $u > 0$ and $\lambda > 0$, the argument of the logarithm is strictly positive and finite. Thus, $h(\mathbf{N}) < \infty$.
    \item \textbf{Lower Bound:} The lower bound is given by $\underline{h}(u) = \frac{p}{2} \log(2\pi e \sigma^2) + \frac{p}{2} \mathbb{E}[\log T]$. By Lemma~\ref{lem:log_moment} established in this Appendix, we have $\mathbb{E}[|\log T|] < \infty$, which implies $\mathbb{E}[\log T]$ is finite. Consequently, $h(\mathbf{N}) > -\infty$.
\end{itemize}
Since both bounds are finite, we conclude that $|h(\mathbf{N})| < \infty$.
\end{IEEEproof}


\subsection{Finiteness of the Mutual Information between NDFHL and the Mixing Variable}\label{subsec:finite_MI}The validity of the lower bound expression in Theorem~\ref{thm:refined_scaling} relies on the finiteness of the mutual information between the mixing time and the spatial noise. We formally establish this property here.\begin{lemma}\label{lem:finite_MI}Let $\mathbf{N} \sim \mathrm{NDFHL}^{(d)}(\lambda, u)$ with $u > 0$. The mutual information $I(T; \mathbf{N})$ between the first-hitting time $T$ and the location $\mathbf{N}$ is finite.\end{lemma}\begin{IEEEproof}
Since $\mathbf{N}$ follows the $\mathrm{NDFHL}^{(d)}$ distribution, we can directly invoke the entropy decomposition and bounds established in Lemma~\ref{lem:entropy_bounds} (Section~\ref{subsec:entropy_analysis}). Specifically, the mutual information is bounded by the Jensen gap of the mixing time:
\begin{equation}
0 \le I(T; \mathbf{N}) \le \frac{p}{2} \left( \log \mathbb{E}[T] - \mathbb{E}[\log T] \right).
\end{equation}
For $u > 0$, the first-hitting time $T$ follows an Inverse Gaussian distribution with a finite mean $\mathbb{E}[T] = \lambda/(u\sigma^2) < \infty$. Furthermore, Lemma~\ref{lem:log_moment} in this Appendix established that $\mathbb{E}[|\log T|] < \infty$, which implies that $\mathbb{E}[\log T]$ is finite. Since the upper bound is finite, it follows that $0\le I(T; \mathbf{N}) < \infty$.
\end{IEEEproof}


\section{The GH and NIG-Type Distribution Families}
\label{app:GH_NIG}

\begin{table*}[ht]
\centering
\caption{Relationship between the GH, NIG, and NDFHL distribution families.}
\label{tab:GH_NIG_NDFHL}
\begin{tabular}{lccc}
\hline
 & GH & NIG & NDFHL \\ \hline
Support & $\mathbb R$ (or $\R^{d-1}$ for elliptic GH) & $\mathbb R$ (or $\R^{d-1}$ for elliptic NIG)& $\mathbb R^{d-1}$ \\[2pt]
Primary role & Umbrella family & GH subclass & Planar exit law \\[2pt]
Key parameters &
$(\nu_{\mathrm{gh}},\alpha,\beta,\delta,\mu)$ &
$(\alpha,\beta,\delta,\mu)$ &
$(\lambda, u)$ \\[2pt]
Bessel kernel &
$K_{\nu_{\mathrm{gh}}-\frac12}$ &
$K_{1}$ &
$K_{\frac{d}{2}}$ \\[2pt]
Stochastic origin &
Variance--mean mixture &
Gaussian + IG mixture &
Boundary hitting \\[2pt]
Exact coincidence &
-- &
GH with $\nu_{\mathrm{gh}}=-\tfrac12$ &
Multivariate Cauchy (under limit $u\to 0$) \\ \hline
\end{tabular}
\end{table*}

\subsection{The generalized hyperbolic family}

The GH family, originally introduced by Barndorff-Nielsen \cite{BarndorffNielsen:1977}, forms a broad and flexible class of absolutely continuous distributions, parameterized in the univariate case by $(\nu_{\mathrm{gh}},\alpha,\beta,\delta,\mu)$, where $\nu_{\mathrm{gh}} \in \mathbb{R}$ is the index parameter, while $\alpha>|\beta|$ and $\delta>0$. A commonly used representation of the univariate GH probability density function is
\begin{equation}
\label{eq:GH_pdf}
\begin{aligned}
f_{\mathrm{GH}}(x)
&=
\frac{(\gamma/\delta)^{\nu_{\mathrm{gh}}}}
{\sqrt{2\pi}\,K_{\nu_{\mathrm{gh}}}(\delta\gamma)}
\\
&\quad \times
\frac{
K_{\nu_{\mathrm{gh}}-\frac12}\!\Big(
\alpha\,\sqrt{\delta^2+(x-\mu)^2}
\Big)
}{
\Big(\sqrt{\delta^2+(x-\mu)^2}/\alpha\Big)^{\frac12-\nu_{\mathrm{gh}}}
}
\,\exp\!\bigl(\beta(x-\mu)\bigr).
\end{aligned}
\end{equation}
where $\gamma:=\sqrt{\alpha^2-\beta^2}$ and $K_\nu(\cdot)$ denotes the modified Bessel function of the second kind.

Beyond the univariate setting, the GH family admits a natural \emph{elliptically contoured} multivariate extension \cite{McNeil:2015}. A random vector $\mathbf{X} \in \mathbb{R}^p$ is said to follow an elliptic GH distribution if it admits the stochastic representation
\begin{equation}
\label{eq:elliptic_GH_mix}
\mathbf{X}
=
\boldsymbol{\mu}
+
W\,\boldsymbol{\beta}
+
\sqrt{W}\,A\,\mathbf{Z},
\end{equation}
where $\mathbf{Z}\sim\mathcal N(\mathbf{0}, \mathbf{I}_p)$, $A A^\mathsf T = \Sigma$ is a positive definite dispersion matrix, and the mixing variable $W$ follows a generalized-inverse-Gaussian (GIG) distribution \cite{jorgensen2012}. The resulting density depends on $\mathbf{x}$ only through the Mahalanobis radius $\sqrt{(\mathbf{x}-\boldsymbol{\mu})^\mathsf T \Sigma^{-1}(\mathbf{x}-\boldsymbol{\mu})}$ \cite{Mahalanobis:1936}, reducing to \eqref{eq:GH_pdf} when $p=1$.

\subsection{The normal-inverse-Gaussian family}

The NIG family \cite{BarndorffNielsen:1997} corresponds to the special case $\nu_{\mathrm{gh}}=-\tfrac12$ of the GH family. In one dimension, the probability density function simplifies to
\begin{equation}
\label{eq:NIG_pdf}
f_{\mathrm{NIG}}(x)
=
\frac{\alpha\delta}{\pi}
\,
\frac{
\exp\bigl(\delta\gamma+\beta(x-\mu)\bigr)
}{
\sqrt{\delta^2+(x-\mu)^2}
}
\,
K_1\!\left(
\alpha\sqrt{\delta^2+(x-\mu)^2}
\right),
\end{equation}
where $\gamma=\sqrt{\alpha^2-\beta^2}$. The NIG family has semi-heavy (exponentially damped) tails and admits a convenient IG variance--mean mixture representation, rendering it particularly compatible with diffusion-driven observables.

In $\mathbb{R}^p$, the elliptic (or isotropic) NIG distribution is defined analogously by setting $\nu_{\mathrm{gh}}=-\tfrac12$ in the multivariate GH construction. In the symmetric case $\boldsymbol{\beta}=\mathbf{0}$, the resulting density takes the radial form:
\begin{equation}
\label{eq:vector_NIG_pdf}
f_{\mathrm{NIG}}^{(p)}(\mathbf{x})
\propto
\frac{
K_{\tfrac{p+1}{2}}\!\left(
\alpha\sqrt{\delta^2+\|\mathbf{x}-\boldsymbol{\mu}\|^2_{\Sigma^{-1}}}
\right)
}{
\left(
\sqrt{\delta^2+\|\mathbf{x}-\boldsymbol{\mu}\|^2_{\Sigma^{-1}}}
\right)^{\tfrac{p+1}{2}}
},
\end{equation}
where $\|\mathbf{y}\|^2_{\Sigma^{-1}} = \mathbf{y}^\mathsf T \Sigma^{-1} \mathbf{y}$. The order of the Bessel kernel adjusts automatically with the ambient dimension.

\subsection{Relevance to NDFHL noise family}

The connection between the NDFHL family defined in this paper and the NIG-type distributions becomes transparent at the level of radial structure. Recall that the $\mathrm{NDFHL}^{(d)}(\lambda, u)$ density defined in \eqref{eq:NDFHL_pdf} (where the noise dimension is $p=d-1$) depends on the spatial coordinate $\mathbf{n} \in \mathbb{R}^p$ only through $\sqrt{\|\mathbf{n}\|^2+\lambda^2}$, with kernel:
\begin{equation}
K_{\frac{d}{2}}\!\left(
u\sqrt{\|\mathbf{n}\|^2+\lambda^2}
\right).
\end{equation}
Noting that $\frac{d}{2} = \frac{p+1}{2}$, this coincides (up to normalization) with the isotropic vector NIG density \eqref{eq:vector_NIG_pdf} under the parameter identification:
\begin{equation}
\alpha = u,
\qquad
\boldsymbol{\beta} = \mathbf{0},
\qquad
\delta = \lambda,
\qquad
\boldsymbol{\mu} = \mathbf{0},
\qquad
\Sigma = \mathbf{I}_p.
\end{equation}
Under this identification, the NDFHL family can be viewed as a geometrically induced, isotropic NIG-type distribution in $\mathbb{R}^{d-1}$, whose parameters are fully determined by the normalized drift strength $u$ and the Tx--Rx separation $\lambda$.

While the GH family allows for general skewness and elliptical anisotropy, the isotropic NIG subclass captures the essential semi-heavy-tailed and mixture structures of NDFHL observables.
We adopt the specific terminology ``$\mathrm{NDFHL}$'' primarily to emphasize the explicit dependence on the physical transport parameters (normalized drift $u$ and distance $\lambda$) and its role as a channel noise model, rather than to claim the discovery of a distinct statistical family.

To clarify the relationship between the proposed boundary-induced noise model and existing statistical classifications, Table~\ref{tab:GH_NIG_NDFHL} summarizes the conceptual distinction between these families.

\balance
\bibliographystyle{IEEEtran}
\bibliography{refs-v2}

\end{document}